%% file: two-way-main.tex
\documentclass[runningheads]{llncs}
\input{macro.tex}

\begin{document}
\title{2-Way 1-Clock ATA \& Its Logics: Back To The Future With Alternations}
\author{Shankara Narayanan Krishna \inst{1} \and
Khushraj Madnani\inst{2} \and
Manuel Mazo Jr. \inst{2} \and Paritosh K. Pandya \inst{1}}
\authorrunning{S.N.Krishna et al.}
\institute{IIT Bombay, Mumbai, India \email{krishnas,pandya58}@cse.iitb.ac.in\and Delft University of Technology, Delft, The Netherlands\footnote{This work is partially supported by the European Research Council through the SENTIENT project (ERC-2017-STG \#755953).} \email{k.n.madnani-1,m.mazo}@tudelft.nl}
\maketitle
\oomit{
\begin{abstract} 
In this paper, we study the extension of 1-clock Alternating Timed Automata with ability to read in both forward and backward directions: 2-Way 1-clock Alternating Timed Automata (2-Way 1-ATA). We show that the subclass of 2-Way 1-ATA with reset free loops (2-Way 1-ATA-rfl) is expressively equivalent to a subclass of MSO with Guarded Metric Quantifiers (GQMSO). The emptiness checking problem for 1-ATA-rfl (and hence GQMSO) is undecidable, in general. We propose a generalization of  the classical  non-punctuality  restriction, called non-adjacency, for 2-Way 1-ATA-rfl (and GQMSO) for which the emptiness checking (and satisfiability checking, respectively) becomes decidable.
\end{abstract}
}
\begin{abstract} 
 In this paper, we study the extension of 1-clock Alternating Timed Automata (1-ATA) with the ability to scan the timed behaviour in both forward and backward directions: the 2-Way 1-clock Alternating Timed Automata (2-Way 1-ATA). We show that the subclass of 2-Way 1-ATA with reset free loops (2-Way 1-ATA-rfl) is expressively equivalent to MSO[$<$] extended with Guarded Metric Quantifiers (GQMSO). The emptiness checking problem for 2-Way 1-ATA-rfl (and hence GQMSO) is undecidable, in general. We propose a generalization of the classical non-punctuality restriction, called non-adjacency, for 2-Way 1-ATA-rfl, and also for GQMSO, for which the emptiness (respectively, satisfiability) checking becomes decidable. Non-Adjacent 2-Way 1-ATA-rfl is the first class of timed automata with alternations and 2-wayness for which the emptiness checking is decidable with \emph{elementary} complexity. We also show that 2-Way 1-ATA-rfl, even with the non-adjacent restrictions, can express properties that are not recognizable by 1-ATA. 
\end{abstract}
\section{Introduction and Related Work}
Exploring connections between different logics (e.g. the Kamp Theorem) and also between logics and automata (e.g. the B\"{u}chi Theorems) has been an active and influential area of work. Such connections often bring the ability to analyze logical questions algorithmically. Unfortunately, it has been challenging to find such tight connections between numerous timed logics and timed automata which have been proposed in the literature.

1-way 1-clock Alternating Timed Automata (1-ATA) were proposed as a Boolean closed model of timed languages with decidable emptiness. These were used to show the decidability of the future fragment of real-time logic $\mtl[\until]$ (see \cite{Ouaknine05} \cite{LW05})\footnote{These results are proved for automata and logics over finite timed words with point-wise interpretation. We shall also follow this interpretation in the current paper.}.  However, the logic was not expressively complete for these automata. Exploring connections between real-time classical and temporal logics,
Rabinovich \cite{rabin} as well as Hunter \cite{hunter} showed that logic $\mitl[\until, \since]$ extended with Pnueli modalities has the same expressive power as logic Q2MLO.  The current authors \cite{KKP18} defined a more expressive
and yet decidable extension of $\mtl[\until]$ called $\regmtl$, and showed that this is expressively equivalent to the subclass of 1-ATA where all loops are reset free (1-ATA-rfl). Moreover these are expressively equivalent to a future time logic QkMSO. 

The current paper explores a major extension of these results to logics and automata with both future and past. We show that the 2-Way extension of 1-ATA-rfl (2-Way 1-ATA-rfl) is expressively equivalent to an extension of MSO[$<$] with Guarded Metric Quantifiers (GQMSO).  The latter is a versatile and expressive logic, allowing properties of real-time systems to be defined conveniently. The use of Guarded Metric Quantifiers appeared in the pioneering formulations of logics QMLO and Q2MLO by Hirshfeld and Rabinovich \cite{rabin} and it was further explored by Hunter \cite{hunter}. We have \textbf{generalized} these to an anchored block of guarded quantifiers with arbitrary depth. This provides the required power to obtain expressive completeness. 

To show the reduction from GQMSO to  2-Way 1-ATA-rfl (and vice versa), the proof factors via a recently proposed extension of  MTL with ``\emph{Pnueli-Automata Modalities}''. This logic has been called Pnueli Extended Metric Temporal Logic (PnEMTL) \cite{KKMP21}. Hence, as our first main result we show, through effective reductions, the exact expressive equivalence $\equiv$ of the following:
\begin{equation}
\label{eqn:main1}
\text{2-way 1-ATA-rfl} ~~\equiv~~ \pnregmtl ~~ \equiv~~ \text{GQMSO}
\end{equation}
The readers may note the conceptual similarity of these results to the celebrated Kamp and B\"{u}chi Theorems.
Unfortunately, the full $\pnregmtl$, being a syntactic extension of $\mtl[\until,\since]$, is clearly undecidable. Hence, emptiness checking and satisfiability checking for both 2-Way 1-ATA-rfl and GQMSO are undecidable. 

In \cite{KKMP21}, we proposed a novel generalization of the non-punctuality condition of $\mitl$ to a \textbf{non-adjacency} condition and showed that the non-adjacent fragments of both $\pnregmtl$ as well as 1-$\tptl[\until,\since]$ have decidable satisfiability with EXPSPACE-complete complexity.  

As our second contribution we define the non-adjacency condition, suitably applied to 2-way 1-ATA automata and the logic GQMSO. We observe that the effective reductions between these formalisms and PnEMTL preserve this non-adjacency. From the previously established EXPSPACE-complete decidability of non-adjacent $\pnregmtl$ (see \cite{KKMP21}), it follows that emptiness of non-adjacent $\toatarfl$ as well as the satisfiability of  non-adjacent GQMSO are decidable. In fact, the former is EXPSPACE-complete. We also show that Non-Adjacent 2-Way 1-ATA-rfl can express properties that cannot be specified in 1-ATA, making their expressive powers incomparable.

To the best of our knowledge, this gives the first subclass of 2-way Alternating Timed Automata which has an  elementary complexity for emptiness checking. In the past, Alur and Henzinger have explored 2-way deterministic timed automata with bounded reversals (Bounded 2DTA) and shown that their non-emptiness is decidable with PSPACE complexity \cite{AHbacktofuture}. Ouaknine and Worrell  as well as Lasota and Walukiewicz  \cite{Ouaknine05} \cite{LW05} showed that emptiness checking of 1-ATA is decidable 
 with non-primitive recursive complexity over finite words and undecidable over infinite timed words. 
 Abdulla \emph{et al} \cite{ADOQW08} showed that  generalizing  1-ATA, by allowing  $\epsilon$-transitions, 2-wayness or omega words leads to undecidability of the emptiness checking problem. Thus,  our model  {\bf non-adjacent 2-Way 1-ATA with reset free loops}, is quite delicately poised. The expressively complete and decidable logic   {\bf Non-adjacent GQMSO} can be seen as a powerful decidable generalization of Hirshfeld and Rabinovich's  Q2MLO \cite{rabin} \cite{rabinovichY}.  Figure \ref{fig:intro} highlights the place of 2-Way extensions studied in the literature amongst the other studied variants of 1-ATA and logics in terms of expressiveness.
 \begin{figure}
     \centering
     \includegraphics[scale = 0.19]{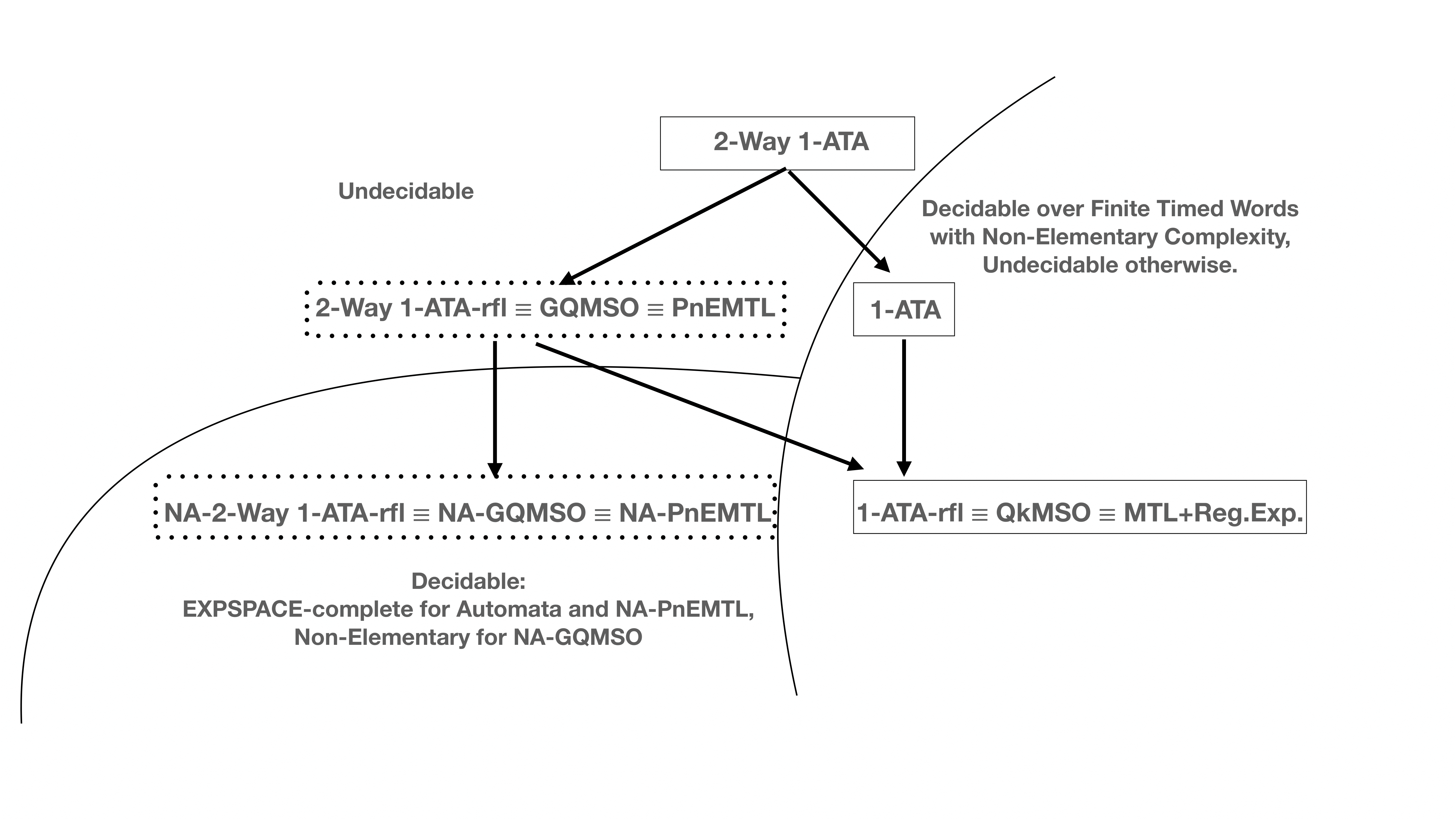}
     \caption{Comparison of expressive power and decidability of some variants of Alternating Timed Automata. An arrow from class A to B implies A is equally or more expressive than B. Classes (and the equivalences) within  dotted boxes are contributions of this paper. }
     \label{fig:intro}
 \end{figure}
\input{prelim_update.tex}
\input{decidable_subclasses.tex}

\input{toatapnreg.tex}
\input{discussion.tex}
\bibliographystyle{plain}
\bibliography{papers}

\newpage 
\appendix
\input{appendix.tex}

\end{document}

%% file: macro.tex
\usepackage[normalem]{ulem}
\usepackage{multirow}
\usepackage{amssymb}
\usepackage{verbatim}
\usepackage{amsmath,amsfonts}
\usepackage{lineno}
\usepackage{graphicx}
\usepackage{color}
\usepackage{xcolor}
\usepackage{amsmath}
\usepackage{amsfonts}
\usepackage{amssymb}
\usepackage{tikz}

\usetikzlibrary{calc}
\tikzstyle{RectObject}=[rectangle,fill=white,draw,line width=0.5mm]
\tikzstyle{line}=[draw]
\tikzstyle{arrow}=[draw, -latex]

\newcommand{\fregm}{\mathcal{F}}
\newcommand{\sregm}{\mathcal{P}}
\newcommand{\fregk}{\mathcal{F}^{k}}
\newcommand{\sregk}{\mathcal{P}^{k}}

\newcommand{\fregkm}{\mathcal{F}^{k}_{\mathsf{I_1,\ldots, I_k}}}
\newcommand{\sregkm}{\mathcal{P}^{k}_{\mathsf{I_1,\ldots, I_k}}}

\newcommand{\sreg}{\mathcal{P}}

\newcommand{\until}{\mathsf{U}}
\newcommand{\since}{\mathsf{S}}

\newcommand{\pnregmtl}{\text{PnEMTL}}
\newcommand{\regmtl}{\mathsf{RatMTL}}

\newcommand{\mtl}{\mathsf{MTL}}
\newcommand{\tptl}{\mathsf{TPTL}}
\newcommand{\mitl}{\mathsf{MITL}}
\newcommand{\emitl}{\mathsf{EMITL}}
\newcommand{\toata}{\text{2-way 1-ATA}}
\newcommand{\toatarfl}{\text{2-way 1-ATA-rfl}}
\newcommand{\ltl}{\mathsf{LTL}}
\newcommand{\mso}{\text{MSO}}

\newcommand{\R}{\mathbb{R}}

\newcommand{\re}{\mathsf{A}}
\newcommand{\lem}{\vdash}
\newcommand{\rem}{\dashv}

\newcommand{\tran}{\mathsf{\delta_{tr}}}
\newcommand{\A}{\mathcal{A}}
\newcommand{\B}{\mathcal{B}}

\newcommand{\cC}{\mathcal{C}}
\newcommand{\C}{\mathbb{C}}
\newcommand{\G}{\mathcal{G}}
\newcommand{\I}{\mathcal{I}}
\newcommand{\J}{\mathcal{J}}
\newcommand{\h}{\mathbf{h}}
\newcommand{\s}{\mathcal{S}\mathsf{uccessor}}
\newcommand{\seq}{\mathsf{seq}}
\newcommand{\rev}{\mathsf{Rev}}

\newcommand{\absa}{{ABS(\A)}}

\newcommand{\succs}{\mathsf{Succ}}
\newcommand{\Qq}{\mathcal{Q}}

\newcommand{\tforall}{\overline{\forall}}
\newcommand{\texists}{\overline{\exists}}

\newcommand{\intintervaln}{\mathcal{I}_{\mathsf{+}}}
\newcommand{\intintervalneg}{\mathcal{I}_{\mathsf{-}}}
\newcommand{\intinterval}{\mathcal{I}_{\mathsf{+,-}}}

\newcommand{\boundaryint}{\mathsf{Boundary}}
\newcommand{\anch}{\mathsf{anch}}
\newcommand{\first}{\mathsf{first}}
\newcommand{\last}{\mathsf{last}}
\newcommand{\Ii}{\mathcal{I}}

\newcommand{\col}{\mathsf{Collapse}}
\newcommand{\normalize}{\mathsf{Norm}}
\newcommand{\Clcap}{\mathsf{CL}}
\newcommand{\md}{\mathsf{MD}}
\newcommand{\mtd}{\mathsf{MtD}}

\newcommand{\free}{\text{free}}
\newcommand{\bound}{\text{bound}}

\newcommand{\oomit}[1]{}
\newcommand{\init}{\text{init}}
\newcommand{\bseq}{\mathsf{BSequence}}
\newcommand{\Qseq}{\mathsf{Qseq}}
\newcommand{\qseq}{\mathsf{qseq}}
\newcommand{\Ac}{\mathcal{A}^c}

%% file: prelim_update.tex
\section{Preliminaries}
\label{sec:prelims}
Let $\Sigma$ be a finite set of propositions, and let  $\Gamma = 2^{\Sigma} \setminus \{\emptyset\}$. A (finite) word over $\Sigma$ is a (finite) sequence $\sigma = \sigma_1 \sigma_2 \ldots \sigma_n$, where $\sigma_i \in \Gamma$. A (finite) timed word $\rho$ over $\Sigma$ is a (finite) sequence of pairs in $\Gamma \times \R_{\geq 0}$;  
$\rho = (\sigma_1, \tau_1) \ldots (\sigma_n, \tau_n) \in (\Gamma \times \R_{\geq 0})^*$ 
where $\tau_1=0$ and $\tau_i \leq \tau_j$ for all $1 \leq i \leq j \leq n$.  The $\tau_i$ are called time stamps. 
For a timed or untimed word $\rho$, let $dom(\rho) = \{i ~|~ 1 \le i \le |\rho|\}$, where $|\rho|$ denotes length of $\rho$.  
Given a (timed) word $\rho$ and $i \in dom(\rho)$, a pointed (timed) word is the pair $\rho, i$. 
The set of all timed words over $\Gamma$ is denoted by $T\Gamma^*$. 
Let $\intintervaln$ ($\intintervalneg$)  be the set of open, half-open or closed time intervals containing real numbers, such that the end points of these intervals are in $\mathbb{N} \cup \{0,\infty\}$ (($\mathbb{Z} \cup \{-\infty\}) \setminus \mathbb{N}$, respectively). Let $\intinterval = \intintervaln \cup \intintervalneg$.
For $\tau {\in} \R$ and interval  $\langle a, b\rangle$, with $\langle {\in} \{(,[\}$ and $\rangle \in \{ ],) \}$, 
${\tau+\langle a, b\rangle}$ stands for the interval 
${\langle \tau+a, \tau+b\rangle}$.
\subsection {Anchored Interval Word Abstractions}
Let $I_{\mu}\subseteq \intinterval$. 
An $I_{\mu}$-interval word over $\Sigma$ is a word $\kappa$ of the form 
$a_1 a_2  \dots a_n \in 
(2^{\Sigma \cup \{\anch\} \cup  I_\mu})^*$. There is a unique $i \in dom(\kappa)$ called the \emph{anchor} of $\kappa$ such that $\anch \in a_i$ (also denoted as $\anch(\kappa)$).
Let $J$ be any interval in $I_\mu$. We say that a point $i \in dom(\kappa)$ is a $J$-time restricted point if and only if, $J \in a_i$. 
A point $i$ is called a \emph{time restricted point} if and only if either $i$ is $J$-time restricted for some interval $J$ in $I_\mu$ or $\anch \in a_i$. 
\\\noindent \textbf{From $I_\mu$-interval words to timed languages}. Given an $I_\mu$-interval word $\kappa=a_1 \dots a_n$ over $\Sigma$ and 
a timed word $\rho=(b_1, \tau_1)\dots (b_m, \tau_m)$, 
the pointed timed word $\rho, i$ is consistent with $\kappa$ iff 
$dom(\rho){=}dom(\kappa)$, $i{=}\anch(\kappa)$, and, for all $j\in dom(\kappa)$, $b_j=a_j\cap \Sigma$ and for $j \ne i$, $I \in a_j \cap I_\mu$ implies 
$\tau_j - \tau_i \in I$. 
Intuitively, each point $j$ of $\kappa$ does the following. (i) It stores the set of propositions that are true at point $j$ of $\rho$ and (ii) It  also stores the set of intervals $\I \subseteq I_\mu$ such that the time difference between $\anch(\kappa)$ and $j$ of $\rho$ lies within $\bigcap \I$, thus abstracting the time differences from the anchor point($i$) using some set of intervals in $I_\mu$.
We denote the set of all the pointed timed words consistent with a given interval word $\kappa$ as $\mathsf{Time(\kappa)}$. Similarly, given a set $\Omega$ of $I_\mu$ interval words, $\mathsf{Time(\Omega)}{=}\bigcup \limits_{\kappa \in \Omega} (\mathsf{Time(\kappa)})$. 
\begin{example}
Let 
$\kappa{=}\small{\{a,b, (-1,0)\} \{b, (-1,0)\} 
\{a, \anch\} \{b,[2, 3]\}}$ be an interval \\word over the set of intervals 
$\{(-1,0),[2,3]\}$. So, $\anch(\kappa)=3$. 
For timed words
 $\rho{=}\small{(\{a,b\}, 0)(\{b\}, 0.5)(\{a\}, 0.95) (\{b\}, 3)}$, $\rho'{=}\small{(\{a,b\},0)(\{b\},0.8)(\{a\},0.9)(\{b\},3)}$.
 $\rho,3$ and $\rho',3$ are  consistent with $\kappa$.  For $\rho'' {=} (\{a,b\}, 0) (\{b\}, 0.5) (\{a\}, 1.1) (\{b\}, 3)$,  $\rho'', 3$ is not consistent with $\kappa$ as $\tau_1-\tau_3 \notin (-1,0)$ (and also $\tau_4-\tau_3 \notin [2,3]$).
\end{example}
Note that the ``consistency relation'' is a many-to-many relation. For set of intervals $I_\mu$, a pointed timed word $\rho,i$ can be consistent with more than one $I_\mu$-interval word and vice versa. Full technical details on interval words can be found in the Appendix \ref{app:interval word} and \cite{KKMP21}.
\subsection {MSO with guarded metric quantifiers, GQMSO}
We define a real-time logic GQMSO which is interpreted over timed words. 
It includes $\mso[<]$ over words with respect to some alphabet $\Sigma$. This is extended with a notion of time constraint formula $\psi(t)$, 
where $t$ is a free first order variable.
All variables in our logic range over positions in the timed word and not over time stamps (unlike continuous interpretation of these logics). 
There are two sorts of formulae in GQMSO which are mutually recursively defined : 
$\mso^{\mathsf{UT}}$ and $\mso^{\mathsf{T}}$ (where $\mathsf{UT}$ stands for untimed and $\mathsf{T}$ for timed). 
An $\mso^{\mathsf{UT}}$ formula $\phi$  has no real-time constraints except for the time constraint subformula 
$\psi(t) \in \mso^{\mathsf{T}}$. A formula $\psi(t)$ has {only one} free variable $t$ (called anchor), which is a first order variable. 
$\psi(t)$ is defined as a block of real-time constrained quantification applied to a GQMSO formula with no free second order variables; it has the form $\Qq_1t_1. \Qq_2t_2. \dots \Qq_jt_j. ~
\phi(t,t_1,\ldots t_j)$  where $\phi \in \mso^{\mathsf{UT}}$. All the metric quantifiers in the quantifier block constrain their variable relative only to the anchor $t$. The precise syntax follows below.\footnote{In \cite{KKP18}, a similar logic called QkMSO was defined. QkMSO had yet another restriction: it can only quantify positions strictly in the future, and hence was not able to express past timed specifications.}
\\\textbf{Remark}: This form of real time constraints in first order logic were pioneered by Hirshfeld and Rabinovich \cite{rabin} in their logic Q2MLO (with only non-punctual guards) and its punctual extension was later shown to be expressively complete to FO[$<,+1$] by Hunter \cite{hunter} over signals. Here we extend the quantification to an \textbf{anchored block of quantifiers}.

We have a two sorted logic consisting of $\mso^{UT}$ formulae $\phi$ and time constrained formulae $\psi$. Let $a \in \Sigma$, and let $t,t'$ range over first order variables, while $T$ range over second order variables. The syntax of $\phi \in \mso^{\mathsf{UT}}$ is given by: 
\\$ t{=}t'~|~t{<}t'~|~Q_a(t)~|~T(t) \mid \phi {\wedge} \phi~|~{\neg} \phi~|~
 {\exists} t. \phi~|~{\exists}  T \phi~ |~\psi(t)$. \\
Here, $\psi(t) \in \mso^{\mathsf{T}}$
is a time constraint formula whose syntax and semantics are given a little later. A formula in $\mso^{\mathsf{UT}}$ with first order free variables $t_0,t_1, \ldots t_k$ and second-order free variables $T_1, \ldots, T_m$ is denoted $\phi(t_0,\ldots t_k,T_1, \ldots, T_m)$. 
The semantics of such formulae is as usual. Let $\rho = (\sigma_1, \tau_1) \ldots (\sigma_n, \tau_n)$ be a timed word over $\Sigma$. Given $\rho$, positions $i_0, \ldots, i_k$ in $dom(\rho)$, and sets of positions $A_1, \ldots, A_m$ with $A_i \subseteq dom(\rho)$, we define 
\\$\rho,(i_0,i_1,\ldots,i_k,A_1,\ldots,A_m) {\models}\phi(t_0 ,t_1, \ldots t_k,T_1, \ldots, T_m)$
inductively in $\mso[<]$.  
\begin{itemize}
 \item  $(\rho,i_0,\ldots,i_k,A_1,\ldots,A_m) {\models} $ $t_x {<} t_y$ 
 iff $i_x {<} i_y$,  
 \item $(\rho,i_0,\ldots,i_k,A_1,\ldots,A_m) {\models} $ 
 $Q_a(t_x)$ iff $a {\in} \sigma_{i_x}$,
 \item  $(\rho,i_0,\ldots,i_k,A_1,\ldots,A_m) {\models} $ $ T_j(t_x)$ iff $i_x {\in} A_j$, 
 \item $(\rho,i_0,\ldots,i_k,A_1,\ldots,A_m) {\models} $ $\exists t'.\phi(t_0, \ldots t_k,t',T_1, \ldots, T_m)$ iff  \\ $(\rho,i_0,\ldots,i_k,i',A_1,\ldots,A_m) \models \phi(t_0, \ldots t_k,t',T_1, \ldots, T_m)$ for some $i' {\in} dom(\rho)$. 
\end{itemize}
The {\bf time constraint} formula $\psi(t) \in \mso^{\mathsf{T}}$ has the form:
\\$\Qq_1t_1. \Qq_2t_2. \dots \Qq_jt_j. ~
\phi(t,t_1,\ldots t_j)$ 
where $t_1, \ldots, t_j$ are first order variables and $\phi \in \mso^{\mathsf{UT}}$. Each quantifier 
$\Qq_x t_x$ has the form $\texists t_x \in t+ I_x$ or
$\tforall t_x \in t+ I_x$ for a time interval  $I_x \in \intinterval$. 
$\Qq_x$ is called a metric quantifier. Note that each metric quantifier constrains its variable only relative to the anchor variable $t$. Moreover, $\psi(t)$ has no free second order variables.  The semantics of such an anchored metric quantifier is as follows. Let 
\\$(\rho,i_0) \models \texists t_1 \in t{+}I. \phi(t,t_1,\ldots t_j)$ iff $\begin{Bmatrix}
\text{there exists } i_1 \text{ such that }\tau_{i_1} \in \tau_{i_0}+I 
\text{ and,}\\ 
(\rho,i_0,i_1 \ldots i_j) \models  \phi(t,t_1,\ldots, t_j)
\end{Bmatrix}$,
\\$(\rho,i_0) \models \tforall t_1 \in t{+}I. \phi(t,t_1,\ldots t_j)$ iff $\begin{Bmatrix}
\text{for all } i_1 \text{ such that }\tau_{i_1} \in \tau_{i_0}+I 
\text{ implies,}\\ 
(\rho,i_0,i_1 \ldots i_j) \models  \phi(t,t_1,\ldots, t_j)
\end{Bmatrix}$.
\\ Note that metric quantifiers quantify over positions of the timed word and the metric constraint is applied on the timestamp of the corresponding positions.
Each time a constraint formula in GQMSO has exactly one free variable; variables $t_1,\ldots,t_j$ are called time constrained in $\psi(t)$. If we restrict the grammar of a time constrained formula $\psi(t) \in \mso^{\mathsf{T}}$ to contain only a single metric quantifier (i.e. $\Qq_1t_1. \phi(t,t_1)$) and disallow the usage of second order quantification, we get the logic \textbf{Q2MLO} of \cite{rabinovichY}.
\begin{example} 
Consider a sequence over events $\Sigma =\{a,b\}$ such that from every $a$ there were positive even number of $b$'s in the previous unit interval.
\\$\phi = \forall t. Q_a(t) \rightarrow \psi(t)$ where 
\\$\psi(t) = [ \exists t_{f} \in t+ [-1,0]. \exists t_{l} \in t+[-1,0] \forall t' \in t+[-1, 0]. \gamma(t,t_f,t_l,t')$ where 
\\$\gamma(t,t_f,t_l,t') = t_f \le t' \le t_l \wedge \exists X_{o}. \exists T_e. T_o(t_f) \wedge T_e(t_l) \wedge \forall t_1. \forall t_2.$ \\ $[\{Q_b(t_1) \wedge Q_b(t_2) \wedge \forall t_3. (t_1<t_3<t_2 \rightarrow \neg Q_b(t_3))\} \rightarrow$ \\$\{(X_o(t_1) \wedge \neg X_e(t_1) \wedge X_e(t_2) \wedge \neg X_o(t_2)) \vee (X_e(t_1) \wedge \neg X_e(t_1) \wedge X_o(t_2) \wedge \neg X_o(t_2))\}]$. 
Here  
$\phi$ is a formula of type $\mso^{\mathsf{UT}}$ containing the subformula $\psi(t)$ of type $\mso^T$ which in-turn contains the formula $\gamma(t,t_f,t_l,t')$ of type $\mso^{\mathsf{UT}}$.
\end{example}
Note that, while GQMSO extends classical MSO$[<]$,  GQMSO is not closed under second order quantification:
arbitrary use of second order quantification is not allowed, and its syntactic usage as explained above is restricted to prevent a second order free variable from occurring in the scope of the real-time constraint (similar to \cite{Raskin-Thesis}, \cite{icalp-raskin} and  \cite{Wilke}). For example, $\exists X. \exists t. [X(t) \wedge \texists t' {\in} t{+}(1,2) Q_a(t')]$ is a well-formed GQMSO formula while, $\exists X. \exists t. \texists t' {\in} t{+}(1,2)[Q_a(t')\wedge X(t)]$ is not, since $X$ occurs freely within the scope of the metric quantifier.

\begin{example}
\label{ex:insterr}
We define a language $\mathsf{L_{inst err}}$ over the singleton alphabet $\Sigma = \{b\}$ accepting words satisfying the following conditions:
\\1. One $b$ with timestamp 0 at the first position. (Positions are counted $1,2,3,\ldots$).
\\2. Exactly two points in the interval $(0,1)$ at positions $2$ and $3$ with timestamps called $\tau_2$ and $\tau_3$, respectively.
\\3. Exactly one $b$ in  $[\tau_2+1, \tau_3 +1]$ at some position $p$. Other $b$'s can occur freely elsewhere.
\\The above language was proposed by Lasota and Walukiewicz \cite{LasotaW08} (Theorem 2.8) as an example of language not recognizable by 1-ATA but expressible by a Deterministic Timed Automata with 2 clocks. Let $S(u,v)$ be the FO[$<$] formula specifying  the successor relation (i.e. $u = v+1$). This can be specified as the GQMSO formula $\psi = \psi_1 \land \psi_3$, where 
\begin{enumerate}
    \item Let $Pos_1(t) {=} \neg \exists w. S(t, w)$, $Pos_i(t) {=} \exists t'. S(t,t') \wedge Pos_i(t')$. Hence, $Pos_i(t)$ holds only when $t=i$, where $i \in \{1,2,3,4\}$.
    \item Let $\psi_1{=}\exists t_1. ~Pos_1(t_1) {\land} (\texists t_2 \in t_1+(0,1) . \texists t_3 \in t_1 +(0,1). [Pos_2(t_2) \wedge Pos_3(t_3) \wedge 
\neg \texists t \in t_1+(0,1). Pos_4(t)]$. This states that exactly two positions exist in the initial unit time interval $(0,1)$. Let their time stamps be $\tau_2$ and $\tau_3$.
    \item Let $\psi_2(p) = [~\texists t \in p+[-1,0). Pos_3(t) ~~\wedge~~ \neg \texists t \in p+(-1,0). Pos_2(t)~]$. This states that position $p$ lies within $[\tau_2+1, \tau_3 +1]$.
    \item $\psi_3 ~=~ \exists p . ~[\psi_2(p) ~\land (\forall q. \psi_2 (q) ~\rightarrow (p=q))]$ states that there is exactly one position satisfying property $\psi_2$.
\end{enumerate}
\end{example}
{\bf{Metric Depth}}. The \emph{metric depth} of a formula $\varphi$ denoted ($\mtd(\varphi)$) gives the nesting depth of time constraint constructs and is defined inductively: For atomic formulae $\varphi$,  $\mtd(\varphi)=0$.
$\mtd[\varphi_1 \land \varphi_2]= \mtd[\varphi_1 \vee \varphi_2] = max(\mtd[\varphi_1],\mtd[\varphi_2])$ and $\mtd[\exists t. \varphi(t)]{=}\mtd[\neg \varphi]{=}\mtd(\varphi(t))$. $\mtd[\Qq_1t_1 \ldots \Qq_j t_j \phi]$ $=$ $\mtd[\phi] + 1$.\\For example, the sentence $\forall t_3~\tforall t_1 \in t_3 + (1,2)~ \{Q_a(t_1) {\rightarrow} (\texists t_0 \in t_1 + [1,1]~ Q_b(t_0))\}$ accepts all timed words such that for each $a$ which is at distance $(1,2)$ from some time stamp $t$, there is a $b$ at distance 1 from it. This sentence has metric depth two with time constrained variables $t_0,t_1$.
\subsubsection{GQMSO with Alternation Free Metric Quantifiers (AF-GQMSO)} 
\label{sec:afgqmso}
We define a syntactic fragment of GQMSO, called AF-GQMSO, where all the metric quantifiers in the outermost quantifier block of every MSO$^{\mathsf{T}}$ subformulae are existential metric quantifiers. More precisely, AF-GQMSO is a syntactic fragment of GQMSO where the {\bf time constraint} $\psi(t_0)$ has the form 
$\texists t_1 \in t_0+I_1. \texists t_2 \in t_0 + I_2. \dots \texists t_j \in t_0 +I_j. ~
\phi(t_0,t_1,\ldots t_j)$ with $\phi \in \mso^{\mathsf{UT} }$. Hence,  there  is  no  alternation of metric quantifiers within a block of the metric quantifier. Note that the negation of the timed subformula is allowed in the syntax of GQMSO (and hence AF-GQMSO). Hence,  alternation free $\tforall^*$ formulae can also be expressed using AF-GQMSO. Later, we show that AF-GQMSO is as expressive as GQMSO.

\subsection{Metric Temporal Logic ($\mtl$)}
$\mtl$ is a real-time extension of $\ltl$ where the modalities until ($\until$) and since ($\since$) are guarded with intervals. Formulae of $\mtl$ are built from $\Sigma$ using Boolean connectives and 
time constrained versions  $\until_I$ and $\since_I$ of the standard $\until,\since$ modalities, where 
 $I \in \intintervaln$. 
Intervals of the form $[x,x]$  are called  punctual; a non-punctual interval is one which is not punctual. Formulae in $\mtl$ are defined as follows.  
$\varphi::=a ~|\top~|\varphi \wedge \varphi~|~\neg \varphi~|
~\varphi \until_I \varphi ~|~ \varphi \since_I \varphi$, 
where $a \in \Sigma$ and  $I \in \intintervaln$.    
For a timed word $\rho = (\sigma_1, \tau_1 ) (\sigma_2, \tau_2) \ldots (\sigma_n, \tau_n) \in (\Gamma \times \R_{\geq 0})^*$, 
with $\Gamma=2^{\Sigma}\backslash \emptyset$, 
a position 
$i \in dom(\rho)$, an $\mathsf{MTL}$ formula $\varphi$, the satisfaction of $\varphi$ at a position $i$ 
of $\rho$, denoted $\rho, i \models \varphi$, is defined below. We discuss the time constrained modalities. 
\begin{itemize}
\item $\rho,i\ \models\ \varphi_{1} \until_{I} \varphi_{2}$  $\iff$  $\exists j > i$. 
$\rho,j\ \models\ \varphi_{2}, \tau_{j} - \tau_{i} \in I$, and  $\forall$ $i< k <j$. $\rho,k\ \models\ \varphi_{1}$ , 
\item $\rho,i\ \models\ \varphi_{1} \since_{I} \varphi_{2}$  $\iff$  $\exists j < i$. $\rho,j\ \models\ \varphi_{2}, \tau_{i} - \tau_{j} \in I$, and   $\forall$ $j< k <i$. $\rho,k\ \models\ \varphi_{1}$.
\end{itemize}
The language of an $\mtl$ formula $\varphi$ is defined as $L(\varphi) {=} \{\rho | \rho, 1 \models \varphi\}$. We say that a formula $\varphi$ is \textbf{satisfiable} iff $L(\varphi)\neq \emptyset$.
The subclass of $\mathsf{MTL}$ where punctual intervals are disallowed is called Metric Interval Temporal Logic $\mathsf{MITL}$. As we are using strict semantics of $\until$ and $\since$, $next$ and $previous$ are trivially definable. 
Satisfiability checking is undecidable for $\mtl[\until,\since]$ \cite{AH93} and EXPSPACE-complete for $\mathsf{MITL}$ \cite{AFH96}.
\subsubsection{MTL extended with Automata Modalities}
There have been several attempts to extend the logic $\mtl[\until]$ with regular expression/automaton
modalities \cite{Wilke,KKP17,F18,H19}. Among these, \cite{Wilke} was the first to extend the logic $\mitl$ with automata modalities, called Extended Metric Interval Temporal Logic ($\emitl$). In our very recent work \cite{KKMP21}, we use a generalization of these automata modalities to give the logic Pnueli-Extendend Metric Temporal Logic ($\pnregmtl$). For any Finite Automaton (NFA) $A$, let $L(A)$ denote the language of $A$.

For an alphabet $\Sigma$, the formulae of $\pnregmtl$ have the following syntax: \\
$\varphi::{=}a~|\varphi \wedge \varphi~|~\neg \varphi~| \fregm^k_{I_1,\ldots,I_k} (\re_1,\ldots, \re_{k+1})(S)~|~ 
\sregm^k_{I_1,\ldots,I_k} (\re_1,\ldots,\re_{k+1})(S)$\\ where $a \in \Sigma$, 
$I_1, I_2, \ldots I_k \in \intintervaln$ and $\re_1, \ldots \re_{k+1}$ are automata over $2^S$ where $S$ is a set of formulae from $\pnregmtl$. 

Let $\rho {=} (a_1, \tau_1),\ldots (a_n, \tau_n) \in T\Gamma^*$, $x,y \in dom(\rho)$, $x{\le} y$ and $S {=} \{\varphi_1,\ldots, \varphi_n\}$ be a given set of $\pnregmtl$ formulae. Let $S_i$ be the exact subset of formulae from $S$ evaluating to true 
at $\rho, i$, and let  $\mathsf{Seg^+}({\rho},{x},{y},S)$ and $\mathsf{Seg^{-}}({\rho},{y},{x},S)$
be the  untimed words $S_x S_{x+1} \ldots  S_y$ and $S_y  S_{y{-}1} \ldots  S_x$ respectively.   
Then, the semantics for $\rho,i_0$ satisfying a $\pnregmtl$ formula $\varphi$ is defined recursively as :
\begin{itemize}
\item $\rho,i_0{\models}\fregm^k_{I_1,\ldots,I_k}(\re_1,\ldots,\re_{k+1})(S)$ iff 
 ${\exists} {i_0{ {<} }i_1{<} i_2 \ldots {<} i_k {<} n}$ s.t.\\ $\bigwedge \limits_{w{=}1}^{k}{[(\tau_{i_w} {-} \tau_{i_0} {\in} I_w)}
 \wedge \mathsf{Seg^+}(\rho, i_{w{-}1}+1, i_w,S) {\in} L({\re_w})] \wedge 
 \mathsf{Seg^+}(\rho, i_{k}, n,S) {\in} L({\re_{k+1}}) $\\
\item $\rho,i_0 \models \sregm^k_{I_1,I_2,\ldots,I_k} (\re_1,\ldots,\re_k, \re_{k+1})(S)$ iff  
${\exists} i_0  {>} i_1 {>} i_2 \ldots {>} i_k > 1$ s.t. \\
$\bigwedge \limits_{w{=}1}^{k}[(\tau_{i_0} {-} \tau_{i_w} {\in} I_w)
 \wedge \mathsf{Seg^{-}}(\rho, i_{w{-}1}-1, i_w,S) {\in} L({\re_{w}})] \wedge 
 \mathsf{Seg^{-}}(\rho, i_{k}, n,S) {\in} L({\re_{k+1}})$
 \footnote{Unlike \cite{KKMP21}, we introduce the strict version of modalities, without loss of generality, for technical reasons. This doesn't affect the complexity of satisfiability checking for its non-adjacent fragment.}. 
\end{itemize}
\begin{figure}[h]
\scalebox{0.6}{
    \includegraphics{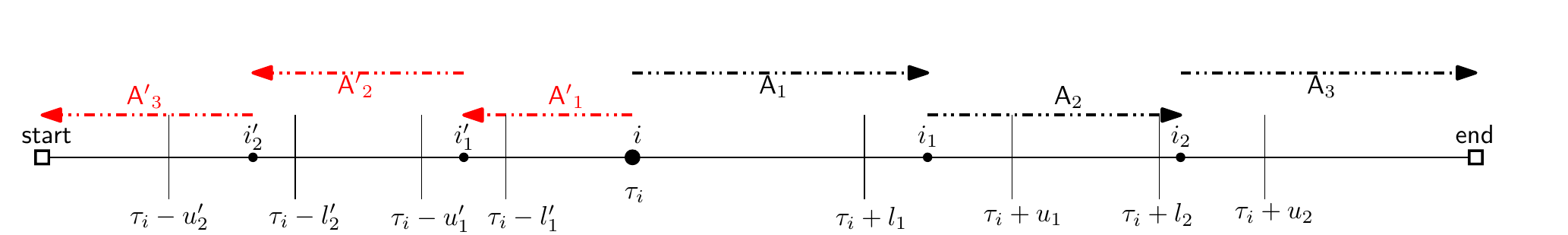}} 
    \caption{Semantics of $\pnregmtl$.$\rho, i{\models}\fregm^{2}_{I_1, I_2}(\mathsf{\re_1, \re_2, \re_3})$ {\&} $\rho, i {\models} \sregm^2_{J_1, J_2}(\mathsf{\re'_1, \re'_2, \re'_3})$ where $I_1 {=} \langle l_1, u_1\rangle, I_2 {=} \langle l_2, u_2\rangle, J_1 {=} \langle l'_1, u'_1\rangle$, $J_2 {=} \langle l'_2, u'_2\rangle$}
    \label{fig:pnregmtl}
\end{figure}
The language of any $\pnregmtl$ formulae $\varphi$ is $L(\varphi) = \{\rho | \rho,1 \models \varphi\}$. 
Given a $\pnregmtl$ formula $\varphi$, its arity is the maximum number of intervals appearing in any $\fregm, \sregm$ modality of $\varphi$. For example, the arity of $\varphi=\fregm^2_{I_1,I_2}(\re_1,\re_2, \re_3)(S_1) \wedge \sregm^1_{I_1} (\re_1,\re_2)(S_2)$ is 2. 
\begin{example}
Consider the formula $\fregm^2_{(1,2)(2,3)}(\{a\}^*\{b\},\{b\}^*\{a\},\{a\}^*)(\{a,b\})$. This formula specifies, that there are sequences of points where $a$ consecutively holds, followed by a sequence of $b$'s again followed by a sequence of $a$'s. Moreover, the first sequence ends within the time interval $(1,2)$ while the second sequence ends within the interval $(2,3)$ from the present point.
\end{example}
 \noindent \textbf{Modal Depth}.
The Modal Depth of a formula $\varphi$, denoted $\md(\varphi)$, is a measure of the nesting of its temporal modalities defined recursively as follows. $\md(a){=} 0$ for any proposition $a$, $\md(\varphi \vee \psi) =  \mathsf{Max}(\md(\varphi), \md(\psi)), \md(\neg \varphi) = \md(\varphi)$,  \\$\md(\mathcal{M}_{I_1,\ldots,I_k}(\re_1,\ldots,\re_{k+1})\langle S\rangle){=}  \mathsf{Max}_{\varphi {\in} S}(\md(\varphi)){+}1$, where $\mathcal{M} \in \{\fregk,\sregk\}$.

\subsection {Expressive Completeness and Strong Equivalence}
Given any specification (formulae or automata) $X$ and $Y$, $X$ is \textit {equivalent} to $Y$ when for any pointed timed word $\rho,i$, $\rho,i \models X \iff \rho,i \models Y$.
We say that a  formalism $\mathcal{X}$ (logic or machine) is \emph{expressively complete} to $\mathcal{Y}$, denoted by $\mathcal{Y} \subseteq \mathcal{X}$, if and only if, for any formulae/automata $X \in \mathcal{X}$ there exists an equivalent $Y \in \mathcal{Y}$. 
$\mathcal{X}$ is said to be \emph{expressively equivalent} to $\mathcal{Y}$, denoted by $\mathcal{X} \cong \mathcal{Y}$ when $\mathcal{X} \subseteq \mathcal{Y}$ and $\mathcal{Y} \subseteq \mathcal{X}$.

\section{Two Way 1-clock Alternating Timed Automata}
\label{sec:toata}
We now define an extension of 1-ATA \cite{Ouaknine05} \cite{LW05}, with ``two wayness''.
Let $\Sigma$ be a finite alphabet. A $\toata$ is a 6 tuple $\mathcal{A}=(\Sigma, Q^+, Q^-, init, \top ,\bot, \delta, \G)$, where 
$Q^+ \cap Q^-  = \emptyset$, $Q = Q^+ \cup Q^-$, and $Q^+$ and $Q^-$ are finite sets of \emph{forward} and \emph{backward} moving locations, respectively. $init \in (Q^+ \cup Q^- \cup \{\top, \bot\})$  is the initial location , $\top$ and $\bot$ are accepting and rejecting locations, respectively. 
Let $x$ denote the clock variable as in 1-ATA (but it can take negative values unlike 1-ATA), and $x \in I$ denote a clock constraint where $I \in \intinterval$. Then $\G$ is a finite set of clock constraints. We say that a real number $\mu$ satisfies a clock constraint $x \in I$, denoted by $\mu \models x \in I$ iff $\mu \in I$.
\\Let $\Sigma' = \Sigma \cup \{\lem, \rem \}$ where $\lem , \rem$ are left and right end markers, respectively. Let $\rho$ be any word over $\Sigma$ with $\tau_{last}$ being the timestamp of the last time point. 
\\Let $Q = Q^+ \cup Q^- $. The transition function is defined as $\delta: Q \times \Sigma' \times \G \rightarrow \Phi(Q')$ where
$Q' = Q\cup \{\top, \bot\})$ and $\Phi(Q')$ is a set of formulae over $Q'$ defined by the grammar as follows. $\varphi ::= \psi \vee \varphi ~|~ \bot $, $\psi::=\psi \wedge \psi~|~q~|~x.q~|~ \top$, where  $q \in Q'$ and $x.q$ is a binding construct  resetting clock $x$ to 0. In other words, $\Phi(Q')$ is a family of positive boolean formulae in Disjunctive Normal Form (DNF) over literals from $Q' \cup \{x.q | q \in Q'\}$.
 
We denote by $\free(\varphi)$, the set of all the locations in $Q$ which do not appear within the 
scope of a reset construct in $\varphi$. Similarly, we denote by $\bound(\varphi)$, the set of all the locations in $Q$ which appear within the scope of a reset construct in $\varphi$. Note that $\free(\varphi)$ and $\bound(\varphi)$ are not necessarily disjoint sets as any location $q$ can be both within and beyond the scope of a reset construct. For example, in $\varphi = q \wedge x.q$, $\free(\varphi) = \bound(\varphi) = \{q\}$.
We define   $\tran(q, a, \mu) = \bigvee \limits_{g \in \G \wedge \mu \models g} \delta(q,a,g)$ \footnote{We define Alternating Finite Automata (AFA) in a similar way as a 7 tuple, $\mathcal{A}=(\Sigma, Q^+, Q^-, init, \top ,\bot, \delta)$. The transition function is $Q \times \Sigma \rightarrow \B(Q\cup \{\top, \bot\})$, where $\B(Q\cup \{\top, \bot\})$ is a Boolean formula (in DNF) over $Q \cup \{\top, \bot \}$, and, $\tran(q,a) = \delta(q,a)$.}.
Given any set of locations $W$, we denote by $W_x$ the set   $\{x.q | q \in W\}$. We apply the following restrictions on transition functions to make sure that the automaton does not ``fall off'' the timed word. For any $q \in Q^-$ and $q' \in Q^+$, $\delta (q,\lem)$ and $\delta (q',\rem)$ are expressions of the form $\Phi(Q^+ \cup \{\top, \bot\})$ and $\Phi(Q^-\cup \{\top, \bot\})$, respectively.

Let $q\in Q$ and $0 \le \h \le m+1$.
A state of a 2-way 1-ATA is either a $\top$ (accepting state) or a $\bot$ (reject state) or a tuple of the form $(q, \mu, \h)$ where $q \in (Q^+ \cup Q^-)$, $\mu$ is a clock valuation and $\h$ is the head position. Formally, a state is an element of $S = ((Q^+ \cup Q^-) \times \mathbb{R} \times (\{0\}\cup \mathbb{N})) \cup \{\top, \bot\}$. A configuration is a set of states.
For any $\toata$, we define a function $\succs$ (which depends solely on the transition function of the given $\toata$) from  a word $\rho$ and a state $s$  to a set of configurations,  $\succs: T\Sigma^* \times S \rightarrow 2^{2^{S}}$, as follows:
\begin{itemize}
\item Let $\rho = (a_1,\tau_1),(a_2,\tau_2),\ldots, (a_m,\tau_m)$. Let $\tau_0 {=} 0$,  $\tau_{m+1}{ = }\tau_m$, $a_0 {= }\lem$ and $a_{m+1} {=} \rem$.
\item $\succs(\rho, \top) = \{\top\}$, $\succs(\rho, \bot) = \{\bot\}$
\item Let $s = (q,\mu, \h)$ be any state, where $0 \le \h \le m+1$. Let $\h' = \h + 1$ if $q\in Q^+$. Otherwise, $\h' = \h - 1$. Let $\mu' = \mu + \tau_{\h'} - \tau_{\h}$.  Let $\tran(q, a, \mu) = \bigvee \limits_{i=1}^{n} (\varphi_i)$ where  $\varphi_i = \top$, $\varphi_i = \bot$ or $\varphi_i$ is of the form $(\bigwedge Q_i \wedge \bigwedge \{x.q|q\in Q'_i\})$ where $Q_i, Q'_i \subseteq Q$. Any configuration $\C \in \succs(\rho,s)$ if and only if there exists $1\le i \le n$,  $\C = \{(q',\mu',\h') | q' \in \free(\varphi_i)\} \cup \{(q',0,\h') | q' \in \bound(\varphi_i)\}$.
    Intuitively, if $q$ is a forward (or backward) moving state then the $\h$ is shifted forward (or backward, respectively) to $\h'$, the valuation of clock $\mu$ is updated to $\mu'$  by adding (or subtracting, respectively) the time delay incurred,  the set of propositions at the new position $\h'$ is read, and non-deterministically, a conjunct (of the DNF) from an  outgoing transition satisfied by the clock valuation $\mu'$ is chosen. The state makes a transition to all the locations appearing in the chosen conjunct simultaneously with the clock valuation as $\mu'$ if a location is free and $0$ if a location is within the scope of a reset.
\end{itemize}
We lift the definition of $\succs$ to configurations, $\s: T\Sigma^* \times 2^{S} \rightarrow 2^{2^{S}}$. Given any two configurations, $\C, \C'$,  $\C' \in \s(\rho, \C)$ if and only if $\C = \{s_1, \ldots s_m\}$ and 
$\C' = \C_1 \cup \ldots \cup \C_m$ such that for every $1\le i \le m, \C_i \in \succs(\rho, s_i)$. Let $\C' \in \s^0(\rho, \C)$ iff  $\C = \C'$. Then we define a function $\s^i$, $\s^i: T\Sigma^* \times 2^{S} \rightarrow 2^{2^{S}}$ such that $\C' \in \s^i(\rho, \C)$ iff  there exists a $\C'' \in \s^{i-1}(\rho, \C) $ such that  $ \C' \in \s(\rho, \C'')$. A configuration $\C$ is accepting if and only if $\C = \{\top\}$. A configuration $\C$ is a rejecting if and only if $\bot \in \C$. 
 Given $\rho$, We say that a configuration $\C$ is $i^{th}$ successor of a configuration $\C'$  with respect to $\rho$ if and only if $\C \in \s^i(\rho, \C')$. A configuration $\C$ is eventually accepting on $\rho$ iff there exists a non-negative integer $n$ such that $\s^n(\rho, \C)= \{\top\}$. 
 
 We say that a pointed timed word $\rho, i \models_\A (q, \mu)$, iff $\{(q, i, \mu)\}$ is eventually accepting on $\rho$.  We say that a pointed timed word $\rho,i$ is accepted by automata $\A$ if and only $\{(init,0,i)\}$ is eventually accepting on $\rho$.
Similarly, a timed word $\rho$ is accepted by automata $\A$ if and only if $\rho,0$ is accepted by $\A$. The language of $\A$, denoted by $L(\A)$, is the set of all timed words accepted by $\A$. To check whether language of a given automaton is empty is called \textbf{emptiness checking}. 
\begin{example}
Consider a 2-Way 1-ATA $A = (\Sigma, Q^+, Q^-, q_0, \top, \bot, \delta, \G)$ where $\Sigma = \{a,b\}$, $Q^+ = \{q_0, q_1\}$, $Q^- = \{p_1\}$, $\G = \{x \in (1,2), x \in (0,1)\}$ and the transition relation is defined as follows. \begin{itemize}
    \item $\delta(q_0, a, x \in (0,1)) = q_0 \wedge x.q_1$, $\delta(q_1, a, x \ne 1) = \delta(q_1, b, x \ne 1) = q_1$,  $\delta (q_1, b, x=1) = \top$. These transitions only allow behaviours where for every occurrence of $a$ within time interval $(0,1)$ there is an occurrence of $b$ exactly after $1$ time units.
    \item $\delta(q_0, \rem) = \top$, $\delta(q_0, b, x \in (1,2)) = x.p_1$, $\delta(p_1, a, x \ne -1) = \delta(p_1, b, x \ne -1) = p_1$,  $\delta (p_1, a, x=-1) = \top$. These transitions only allow behaviours where for every $b$ within time interval $(1,2)$ there was an occurrence of $a$ exactly before $1$ time units.
    \item Moreover the transitions outgoing from $q_0$ make sure that all the $a$'s and $b$'s occur with timestamps in $(0,1)$ and $(1,2)$, respectively.
\end{itemize}
\end{example}
Hence, the above automata accepts words 
whose untimed sequence is of the form $a^nb^n$ for any $n \in \mathbb{N}$. Note that this specification cannot be expressed without the 2-Way extension used here. 
\subsection{Island Normal Form}
We define a normal form for 2-way 1-ATA similar to the normal form of 1-ATA defined in \cite{KKP18}. A 2-way 1-ATA $A = (\Sigma, Q, i, \top ,\bot, \delta,\G)$ is said to be in \emph{Island Normal Form} iff $Q$ can be partitioned into $Q_1, \ldots, Q_n$ and each $Q_i$ has a location called the \emph{header location} $q_{i,r}$ such that:
 \begin{itemize}
 \item For every $a \in \Sigma$ and $q \in Q_i$, $\free(\delta(q,a))\subseteq Q_i\setminus\{q_{i,r}\}$. Hence, all  non-reset transitions outgoing from any location $q\in Q_i$ leads to a non-header location within $Q_i$.
\item For any location $q \in Q$ and $a \in \Sigma$, $\bound(\delta(q,a)) \subseteq \{q_{1,r},\ldots,q_{n,r}\}$. 
\end{itemize}
We call the elements of such partitions as islands. Thus, any transition on which a clock variable is reset, can only lead to the header location of one of the islands.  Therefore, once we enter an island, the only way to leave the island is via a reset transition. Moreover, entry to any island is via reset transition to its header location.
Note that as opposed to the normal form of \cite {KKP18} for 1-ATA, each island here is a reset-free 2-way 1-ATA. 
\subsection{2-way 1-ATA-rfl}
$\A$ is a 2-way 1-ATA-rfl if and only if it satisfies the following:
There is a partial order $(Q_r,\preceq)$ on the header locations (equivalently, on islands $Q_1, \ldots, Q_n$). Moreover, for any location $p \in Q_i$ and a location $q$, 
if $x.q$ occurs in $\delta(p,a)$ for any $a$ (hence $q=q^r_j$) then $q^r_j \prec q^r_i$ ($Q_j \prec Q_i$). Thus,
islands (which are only connected by reset transitions) form a DAG, and
every reset transition goes to a lower level island. Moreover, all transitions within an island are reset-free, but can form cycles. Hence, a cycle can never contain a transition with clock reset.
An island $Q_i$ is a \emph{terminal} island if there is no reset outgoing from any of its states. Hence, all terminal islands are essentially reset free 2-way 1-ATA. Similarly, an island $Q_j$ is said to be initial if its header state, $q^r_j$, is the initial state of $\A$. Note that terminal islands are minimal elements of $\prec$, while the initial island is the maximal element of $\prec$. Note that automata whose island normal form follows the above restrictions are simply automata whose transition graph contains cycles without a reset. The argument for this is similar to that of 1-ATA given in \cite{KKP18}.
\\\textbf{Reset Depth} of any 2-way 1-ATA-rfl $\A$ is the maximum number of reset transitions required to reach a terminal island from the initial island. Hence, the reset depth of a reset free automaton is 0. Similarly, the reset depth of a 2-way 1-ATA-rfl containing only 2 islands is 1.
 \\\noindent\textbf{Boolean Closure of 2-way 1-ATA-rfl} 2-way 1-ATA (rfl) are closed under intersection, union and complementation. The proof of this statement is identical to the case of 1-ATA (Proposition 4 \cite{LW05} or Propositions 7,8 of \cite{Ouaknine05}).
 \begin{lemma}
Any $\toata$ $A$ can be reduced to an equivalent automata in island normal form.
\end{lemma}
The proof is identical to the normalization of 1-ATA described in \cite{KKP18} and \cite{khushraj-thesis}. Hence, without loss of generality we can assume that a given 2-way 1-ATA is in island normal form. 

%% file: decidable_subclasses.tex
\section{Non-Adjacent 2-Way 1-ATA-rfl and GQMSO}
\label{sec:non-adacent classes}

Recently, a generalization of non-punctuality restriction called, \textbf{non-adjacency}, was explored in the context of logics $\tptl$ and $\pnregmtl$ \cite{KKMP21} to gain decidability. We propose similar non-adjacent subclasses of 2-Way 1-ATA and GQMSO in this section and show the decidability for these fragments in section \ref{sec:toatapnemtl}. Any set of intervals $\I$ is said to be non-adjacent iff for any $\mathsf{I_1, I_2} \in \I$, $\mathsf{inf(I_1) \ne sup (I_2)}$. For example, $\{(2,3),(4,5), (2,5)\}$ is  non-adjacent but $\{(0,1), (1,2)\}$ and $\{[1,1]\}$ are adjacent. Note that $[1,1]$ is adjacent to itself and hence it fails the test. Hence, non-adjacency is a generalization of non-punctual restriction of $\mitl$.
\smallskip
\\{\bf Non-Adjacent $\pnregmtl$ (NA-PnEMTL) \cite{KKMP21}}  is defined as a subclass of PnEMTL where every modality $\fregkm$ and $\sregkm$ is such that the given set of intervals $\{\mathsf{I_1, \ldots, I_k}\}$ is a non-adjacent set of intervals. (Note that the same interval can appear several times in the list.)
\subsection{Non-Adjacent 2-way 1-ATA-rfl (NA-2-way-1-ATA-rfl)}
Consider any 2-way 1-ATA $\A = (\Sigma, Q, init, \top ,\bot, \delta)$  with islands $Q_1, \ldots, Q_n$. $\A = (\Sigma, Q, init, \top ,\bot, \delta)$ is non-adjacent iff  the set of all the intervals, $\I_i$, appearing in the outgoing transitions from any location in any island $Q_i$ is non-adjacent. While this class of automata appears to be very restrictive, it can be shown that it can express properties which are not expressible using 1-Way 1-ATA (Theorem \ref{thm:2-way-express}).
\subsection{Non-Adjacent GQMSO (NA-GQMSO)}
Any AF-GQMSO formula $\varphi$ is said to be \textbf{non-adjacent} if and only if for every  subformula $\psi$ of $\varphi$ of the form $\texists t_1 \in t+ I_1 \ldots \texists t_j \in t+I_j \Phi(t,t_1,\ldots,t_j)$,  the set of intervals $\{I_1, \ldots, I_j\}$ is non-adjacent. 
For example, $\texists t_1 \in t_0 +(2,3) \texists t_2. \in t_0+ (3,4) [\exists t <t_0 \wedge \texists t_3 \in t_0+(4,5)]$ is not non-adjacent as intervals $(2,3)$ and $(3,4)$  appear within the same metric quantifier block and are adjacent. On the other hand, $\texists t_1 \in t_0 +(2,3) \texists t_2. \in t_0+ (4,5) [\exists t <t_0 \wedge \texists t_3 \in t_0+(3,4)]$ is non-adjacent as $\{(1,2), (4,5)\}$ is non-adjacent and $\{(2,3)\}$ is non-punctual (and hence non-adjacent to itself). \textbf{Note that the formula in example \ref{ex:insterr} is also a NA-GQMSO formula.}


%% file: toatapnreg.tex
\section{Expressive Equivalences}
\label{sec:toatapnemtl}
\begin{theorem}
\label{thm:toatapnemtl}
(NA-)2-way 1-ATA-rfl $\cong$ (NA-)PnEMTL $\cong$ (NA-)GQMSO. 
\end{theorem}
Before we prove the above theorem, we first observe its implications. On closer examination of the reductions from  2-way 1-ATA-rfl to equivalent NA-PnEMTL here, and from NA-PnEMTL to EMITL$_{0,\infty}$ in \cite{KKMP21} we get Theorem \ref{thm:sat}. Moreover, as a consequence of equivalence of NA-2-Way-1ATA and NA-GQMSO and the example \ref{ex:insterr} we get Theorem \ref{thm:2-way-express}. Appendix \ref{app:na2wayvsata} also gives NA-2-Way-1-ATA-rfl accepting $\mathsf{L_{insterr}}$. The rest of the section is dedicated to proving Theorem \ref{thm:toatapnemtl}.
\begin{theorem}
\label{thm:sat}
Emptiness Checking for 
NA-2-Way 1-ATA-rfl is decidable and EXPSPACE complete.
 Satisfiability 
for 
NA-GQMSO is decidable and non-primitive recursive hard.
\end{theorem}
\begin{theorem}
\label{thm:2-way-express}
$\text{NA{-}2-Way 1{-}ATA{-}rfl}$ can specify properties inexpressible in $\text{1{-}ATA}$.
\end{theorem}
\begin{lemma}
\label{lem:pnemtltoata}
\textbf{PnEMTL$\subseteq$2-Way 1-ATA-rfl.}
\end{lemma}
\textit{ Proof Sketch.}
We apply induction on the modal depth of the formula $\varphi$. For modal depth 0, $\varphi$ is a propositional logic formula. Hence, the lemma trivially holds. For modal depth 1, let $\varphi$ be of the form $\fregk_{I_1, \ldots, I_k} (\re_1, \ldots, \re_{k+1})(\Sigma)$. In the case of $\sregk$ modality, symmetrical construction applies. Moreover, dealing with Boolean operators is trivial as 2-way 1-ATA-rfl is closed under Boolean operations. 
Let $\re_j = (2^{\Sigma}, Q_j, init_j, F_j, \delta_j)$. 
For $a {{\in}} \Sigma$ let Pre$(a,F_j) = \{q| q{\in} Q_j \wedge F_j\cap\delta_j(q,a)\ne \emptyset\}$. Hence Pre$(a,F_j)$ is the set of all the locations in $\re_j$ having an outgoing transition to an accepting state on reading $a$. By semantics, for any timed word $\rho = (a_1, \tau_1) \ldots (a_m,\tau_m)$ and $i_0 {\in} dom(\rho)$, $\rho,i_0 {{\models}} \varphi$ iff there exists a sequence of points $i_1, \ldots, i_k, i_{k+1}$ lying in the strict future of $i_0$ where $i_{k+1} = m$ such that the behaviour of propositions in $\Sigma$ between the segment from $i_j$ to $i_{j+1}$ (excluding $i_j$ and including $i_{j+1}$) is given by the automata $\re_{j+1}$ for any $0\le j \le k+1$. This specification can be expressed using 1-clock non-deterministic timed automata (NTA) $\A = (2^{\Sigma}, Q, init, F, \delta, \G)$, constructed as follows. $Q = Q_1\cup\ldots\cup Q_{k+1}$. $init = init_1$, $F = F_{k+1}$, $\G = \{x {\in} I_1, \ldots, x{\in} I_k\}$, for any $1 \le j \le k+1$ $q {\in} Q_j$, $\delta(q,a) = \delta_j(q,a)$, for any $1\le j \le k$ $q {\in} \text{Pre}(a,F_i)$, $\delta(q,a,x{\in} I_{i}) = init_{j+1}$. 
By the semantics of NTA $\rho, i {{\models}} \varphi$ iff $\A$ reaches an accepting state on reading $\rho$ starting from position $i$. Note that the NTA we constructed is a reset free NTA. In case of $\sreg$ modality we would have a backward moving reset free NTA. Hence, the NTA constructed are ``2-Way 1-ATA-rfl'' with single island. Moreover, $\A$ uses the same set of intervals as $\varphi$. Hence, if $\varphi$ is in NA-PnEMTL, then $\A$ is in NA-2-way 1-ATA-rfl. 
The induction part is identical to the proof of Theorem 13(2) \cite{KKP18} and appears in appendix \ref{app:pnemtltoata}.

\begin{lemma}
\label{lem:toatapnemtl}
\textbf{2-way 1-ATA-rfl $\subseteq$ PnEMTL.}
\end{lemma}
\begin{proof}
Proof is via induction on reset depth of $\A$. 
We give the flow of the construction. 
\begin{enumerate}
\item \textbf{Base Case}: For reset depth 0, $\A$ is a reset free 2-way 1-ATA-rfl. Reduce $\A$ to an untimed 2-way AFA over $\I$ interval words
called $\absa$ by treating guards as symbolic letters,
such that $\rho,i$ is accepted by $\A$ iff $\rho,i {\in} Time(L(\absa))$.
\item Reduce the 2-way AFA, $\absa$, to NFA $A$ over $\I$ interval words using \cite{OG14}.
\item Give the reduction from NFA $A$ over $\I$ interval words to $\pnregmtl$ formula $\varphi$ such that $\rho,i {\in} Time(L(A)){\iff}\rho,i {\models} \varphi$. Hence, $\rho, i {\models} \varphi {\iff} \rho,i {\models}_\A (init, 0)$.
Moreover, if $\I$ is non-adjacent then $\varphi$ is a non-adjacent PnEMTL formula. This step is due to lemma 4 of \cite{KKMP21} and appears in  Appendix \ref{app:nfapnemtl}.
\item \textbf{Induction}: Replace all the lower level islands by witness propositions. Then apply the reduction as in base case. Finally, the witnesses are replaced by subformula equivalent to the corresponding automata. This step is similar to Theorem 13(1) of \cite{KKP18} and appears in Appendix \ref{app:toata-pnemtl}.
\end{enumerate}

We show step 1 here. The rest of the steps rely on \cite{OG14} \cite{KKMP21} \cite{KKP18}.
\\\textbf{Construction of $\absa$}:
Let 
\\$\A = (2^{\Sigma \cup \{\anch\}} , Q^+ \cup \{check, init'\}, Q^-, init', \top ,\bot, \Delta, \G)$. Let $\I$ be the set of intervals appearing in the clock constraints of $\G$. \\
$\absa = (2^{\Sigma \cup \{\anch\} \cup \I},Q^+ \cup \{check, init'\}, Q^-, init', \top ,\bot, \delta)$ such that for any $q \in Q^+ \cup Q^-$, $a \in \Gamma \cup \{\lem, \rem\}$, $\J \subseteq \I$, $\delta(q, a\cup \J) = \bigvee \limits_{I \in \J} \Delta (q, a, x \in I)$, $\delta(q, a\cup  \{\anch\}) = \bigvee \limits_{0 \in I'} \Delta (q, a, x \in I')$.
That is, for every conjunction of outgoing edges from a location $q$ to a set of locations $Q'$ on reading $a \in \Gamma$ with guard $x \in I$  in $\mathcal{A}$, there is a conjunction of outgoing edges from state $q$ to $Q'$ on reading symbol $a \cup \J$ for any $\J \subseteq \I$ and $I \in \J$ or on reading a symbol $a \cup \anch$ if $0 \in I$. Moreover, for any $a \in \Gamma$ and $I \in \I$,
\begin{itemize}
 \item $\delta(init', a\cup \J) = init'$, $\delta(init', a\cup \{anch\}) = check \wedge init$ : Continue to loop till the anchor point is encountered. After reading the anchor point, spawn two locations , $check$ and $init$, simultaneously.
 \item $\delta(init', \rem) = \bot$: If no anchor point is encountered before the head reaches the right end marker, reject the word.
\item $\delta(check, a \cup \J) = check$, $\delta(check, \rem) = \top$, $\delta(check, a \cup \{anch\}) = (\bot, r)$, $\delta(init', \rem) = \bot$: Continue to loop on check after the first encounter of an anchor point. If another anchor point is encountered, reject the word.
\item The above conditions will make sure that a word is accepted only if it has exactly one anchor point and thus is a valid $\I$-interval word.\end{itemize}

The proof of correctness for the above construction (lemma \ref{lem:toataafa}) requires the following proposition. 
Let $\succs_\Delta$ be the successor relation of $\A$. Then 
\begin{proposition}
\label{resetfreevalue}
For any timed word $\rho = (a_1,\tau_1) \ldots, (a_m, \tau_m)$, any point $i \in dom(\rho)$ and any non-negative integer $g$, $\cC \in \succs_\Delta^g(\rho,\{(init, 0, i)\})$ implies that for all $(q,\nu,\h) \in \cC$, $\nu = \tau_\h - \tau_i$.
\end{proposition}
The above proposition can be proved easily by applying induction on $g$. Intuitively, as there is no reset construct, valuation of the clock for any state reachable from the initial state will be equal to the delay from the point where the $\A$ was started. Hence, the clock valuation of all the reachable states will be $\tau_\h -\tau_i$ where $\h$ is the header position of the state. The following lemma proves the language of $\absa$ is a set of interval abstractions of the words accepted by $\mathcal{A}$. Moreover, ``concretizing'' the language of $\absa$ gives back that of $\A$.
\begin{lemma}
\label{lem:toataafa}
Any $\rho,i$ is accepted by $\mathcal{A}$ iff $\rho,i \in Time(L(ABS(\mathcal{A})))$.
\end{lemma}
\begin{proof}
Intuitively, the state $init'$ loops over itself and moves the read header left to right until the head reaches an anchor point. After which it spawns two states simultaneously, $check, init$. The $check$ location checks that there is no other anchor point in the future and thus ensures the uniqueness of the anchor point. On the other hand, $init$ starts imitating the transitions of automata $\mathcal{A}$ in such a way that it precisely accepts interval abstractions of the set of pointed timed words accepted by $\A$. 
 We say that states $(q, \h, \nu)$ of $\A$ and $(q',\h')$ of $\absa$ are equivalent to each other iff $q' = q$, and $\h' = \h$. By construction of $\absa$, for any word $w$, $w, 0 \models_\absa init'$ if and only if $w, i \models_\absa init$ and $w, i \models_\absa check$ (i.e. $w$ is a valid $\I$-interval word). Moreover, any word accepted by $ABS(\A)$ is such that all of its point are time restricted points. 
Rephrasing the lemma as follows:
\\$[\Rightarrow]$ For any $\I$-interval word $w$ and $\anch(w) = i$, if $w,i \models_\absa init$ then for any $\rho,i \in Time(w)$, $\rho,i \models_\A (init, 0)$.
\\$[\Leftarrow]$ For any $\rho,i \models_\A (init, 0, 0)$, there exists an $\I$-interval word $w$ such that $\rho,i \in Time(w)$ and  $w,i \models_\absa init$.
\\
We prove $[\Rightarrow]$, for the converse [$\Leftarrow$] referto  Appendix \ref{app:toatafalem}. Consider any arbitrary $\I$-interval word $w= a'_1\ldots a'_m$, where for some $i \in dom(\rho)$, $a'_i = a_i \cup \{\anch\}$ and for all $j \in dom(w)$, $j \ne i$, $a'_j = a_j \cup \J_j$ for some $\J_j \subseteq \I$ such that $T_j \in \bigcap \J$. Let $a'_0 = \{\lem, \J_1\}$ and $a'_{m+1} = \{\rem, \J_{m}\}$.
Let $\rho, i$ be any pointed timed word in $Time(w)$. Let $a_0 = \lem$, $a_{m+1} =\rem$, $\tau_0 = 0$ and $\tau_{m+1} = \tau_m$ and for any $0 \le j \le m+1$  let $T_j = \tau_j - \tau_i$.
Hence, $\rho  = (a_1,\tau_1) \ldots (a_m, \tau_m)$ and for $0 \le j \le m+1$ $T_j \in \J_j$.

Let $\succs_\Delta$ be the successor relation for $\A$ and $\succs_\delta$ be that of $\absa$. 
By proposition \ref{resetfreevalue}, only states of the form $(q,T_\h, \h)$ are reachable from state $(init,0,i)$. We say that a configuration $C$ of $\absa$ is equivalent to a $\C =  \{(q_1,  T_{\h_1}, \h_1), (q_2,  T_{\h_2}, \h_2), \ldots, (q_n,  T_{\h_n}, \h_n)\}$ of $\A$ iff  $C = \{(q_1, \h_1), \ldots,  (q_n, \h_n)\}$.
Let $s = (q, \h)$ be any state of $\absa$ such that $q \in Q$. Let $s' = (q, T_\h, \h)$ be any state of  $\A$. Let $\h' = \h+1$ if $q \in Q^+$ else $\h' = \h-1$. Let $\J_\h' = \{J'| J' \in \I \wedge T_\h \in J'\}$. In other words, $\J_\h'$ be the maximal subset of $\I$ such that $T_\h \in \bigcap \J'_\h$.
By construction of $\absa$,  $\delta(q, a_\h \cup \J_\h) = \bigvee \limits_{I \in \J_\h} \Delta(q, a_\h, x \in I)$. $\tran(q, a_\h, T_\h) = \bigvee \limits_{I \in \J'_\h} \Delta (q, a_\h, x \in I)$. As $T_\h \in \bigcap \J_\h$, $\J_h \subseteq \J'_h$. Hence, any disjunct of the form $\bigwedge Q'$, where $Q' \subseteq Q^+ \cup Q^- \cup \{\top\}$, that appears in $\delta(q, a_\h\cup \J_\h \})$ also appears in $\tran(q, a_\h, T_\h)$. Hence, for every configuration $C \in \succs_\delta(w,s)$ there exists a $\C \in \succs_\Delta(\rho,s')$ such that $C$ is equivalent to $\C$ (obs 1). 
We show that for any $g \ge 0$ and for any $C \in \s_{\absa}^{g}(w, \{(init,i)\})$, there exists $\C \in \s_{\A}^{g}(\rho,\{(init, 0, i)\})$ such that $C$ is equivalent to $\C$.  
 For $g = 0$ the above statement is trivially true. Assume for $g = k$ the statement is true. Let $C' = \{(q_1,\h_1), \ldots, (q_n,\h_n)\}$ be any configuration of $\absa$ such that $C \in \s_{\absa}^{k}(w, \{(init,i)\})$. Then, by induction hypothesis, there exists a configuration $\C' = \{(q_1,\h_1, \nu_1), \ldots, (q_n,\h_n, \nu_n)\}$ of $\A$ such that $\C' \in \s_{\A}^{k}(\rho, \{(init,0,i)\})$. Note that for every $1 \le j \le n$, by proposition \ref{resetfreevalue}, $\nu_j = T_{\h_j}$. 
 Any configuration $C'' \in  \s_\delta(w, C')$ if and only if $C'' = C'_1 \cup \ldots \cup C'_n$, where for any $1\le j \le n$, $C'_j\in \succs_\delta(q_i, \h_i)$. Let $\C'' = \C'_1\cup \ldots \cup \C'_n$ such that $C'_j$ is equivalent to $\C'_j$ for all $1\le j \le n$.  By (obs 1), for any 
 $1\le j \le n$, $\C'_j \in \succs_\Delta(q_j, \nu_j, \h_j)$. As a result,  $\C'' \in \s_\Delta(\rho, \C')$. Hence, for any configuration $C'' \in \s^{k+1}_{\delta}(w,\{(init,i)\})$ there exists a configuration $\C'' \in \s^{k+1}_{\Delta}(\rho,\{(init,0,i)\})$ such that $C''$ is equivalent to $\C''$. Hence, if $w, i \models_\absa init$ then $\rho, i \models_\A (init, 0)$.
  \end{proof}
  \end{proof}

Lemma \ref{lem:toatapnemtl} and \ref{lem:pnemtltoata} imply that (NA-)PnEMTL $\cong$ (NA-)2-way 1-ATA-rfl.
\begin{lemma}
\label{lem:pnemtlqkmso}
\textbf{PnEMTL $\subseteq$ GQMSO.}
\end{lemma}
The key observation is that conditions of the form $\mathsf{Seg}(i,j,\rho,S) {\in} L(\re)$ can be equivalently  expressed as  MSO[$<$] formulae $\psi_{\re}(i,j)$  using B\"{u}chi Elgot Trakhtenbrot (BET) Theorem \cite{buchi}\cite{Elgot}\cite{Trakhtenbrot}. Replacing the former with latter we get an equivalent GQMSO formula. 
See Appendix \ref{app:pnemtlqkmso} for detailed proof.
We first prove the following theorem which will be essential for proving the converse. Recall the fragment AF-GQMSO in section \ref{sec:afgqmso}.
\begin{theorem}
\label{lem:afqkmsoqkmso}
The subclass AF-GQMSO is expressively equivalent to GQMSO.
\end{theorem}
\begin{proof}
Given formula $\psi(t_0)$, we first eliminate the outermost universal metric quantifier (shown underlined), using four additional existential quantifiers and some non-metric universal quantifiers. We consider intervals of the form $[l,u)$ where $l>0$ for simplicity. This could be analogously generalized for other type of intervals.
Let \\
$\psi(t_0) = 
\texists t_{1} {\in} t_0{+}I_{1}\ldots \texists t_{x-1} {\in} t_0{+}I_{x-1}\underline{\tforall t_x {\in} t_0{+}[l,u)} \mathcal{Q}_{x{+}1} t_{x{+}1}\ldots \mathcal{Q}_j t_j \varphi(t_0, \ldots, t_{j})$. \\ We eliminate $\tforall t_x {\in} t_0{+}[l,u)$ as follows. There are 3 possible cases:
\begin{enumerate}
\item There is no point within $[l,u)$ of $t_0$. In this case, $\psi$ will be vacuously true, 
\\${C_1} = \neg \texists t {\in} t_0 {+} [l,u).$
\item There exists a point within $[l,u)$ and a point in $[u,\infty)$ from $t_0$. In this case, we replace the universal quantifier, with 4 existential metric quantifiers and a non-metric universal quantifier (underlined) as follows.
\\$\begin{matrix}
 &  \texists t_{1} {\in} t_0{+}I_{1}\ldots\underline {\texists t'_{x} {\in} t_0{+}[0,l)}& \\
C_2 =  & \underline { \texists t^-_{x} {\in} t_0{+}[l,u) 
\texists t^{+}_{x}. {\in} t_0{+}[l,u) } &\\ 
 & \underline {\texists t''_{x} \in t_0{+}[u,\infty)} \ldots  \mathcal{Q}_j t_j & 
\end{matrix}
\begin{Bmatrix}
S(t'_{x}, t^-_x) \wedge  \\ 
 (\underline{\forall t_x t^-_x\le t_x \le t^{+}_x)} \rightarrow \varphi(t_0, \ldots, t_j)]\wedge \\ 
S(t^{+}_{x}, t''_{x})\\
\end{Bmatrix}$, 
\\where $S$ is the successor relation definable in MSO[$<$]. 
The formula states that there exists a point $t^-_x$ and $t^{+}_x$ (not necessarily distinct) within $[l,u)$ of $t_0$ such that the previous of $t^-_x$ is in $[0,l)$ and next of $t^{+}_x$ is in $[u, \infty)$. This makes $t^-_x$ and  $t^{+}_x$ as first and last point in interval $[l,u)$, respectively. This implies, $\forall t_x. t^-_x\le t_x \le t^{+}_x$ is equivalent to $\tforall t_x {\in} t_0{+}[l,u)$. 
\item There exists a point within $[l,u)$ and no point within $[u,\infty)$ from $t_0$.
This case is similar to the previous ones. We just need to assert that $t^+_x$ is the last point of the timed word. \\$\begin{matrix}
 &  \texists t_{1} {\in} t_0{+}I_{1}\ldots\underline {\texists t'_{x} {\in} t_0{+}[0,l)}& \\
C_3 =  & \underline { \texists t^-_{x} {\in} t_0{+}[l,u) 
\texists t^{+}_{x}. {\in} t_0{+}[l,u) } &\\ 
 & \ldots \ldots  \mathcal{Q}_j t_j & 
\end{matrix}
\begin{Bmatrix}
S(t'_{x}, t^-_x) \wedge  \\ 
 (\underline{\forall t_x t^-_x\le t_x \le t^{+}_x)} \rightarrow \varphi(t_0, \ldots, t_j)]\wedge \\ 
\neg \exists t > t^+_x\\
\end{Bmatrix}$.
\end{enumerate}

Then,{ $C_1 \lor C_2 \lor C3$} is the required formula.
\end{proof}
\begin{lemma}
\label{lem:qkmsopnemtl}
\textbf{GQMSO $\subseteq$ PnEMTL.}
\end{lemma}
\begin{proof}
It suffices to  show \textbf{AF-GQMSO $\subseteq$ PnEMTL}(thanks to theorem \ref{lem:afqkmsoqkmso}). 
The proof is done via induction on metric depth. Let $\psi(t_0) = \texists t_1 {\in} t_0 + I_1 \ldots \texists t_j {\in} t_0 + I_j. \varphi(t_0,t_1, \ldots, t_j)$ be any AF-GQMSO formula of metric depth $1$. 
\\1) By the semantics of GQMSO, any pointed word $\rho,i {\models} \texists t_1 {\in} t_0 + I_1 \ldots \texists t_j {\in} t_0 + I_j. \varphi(t_0,t_1, \ldots, t_j)$ iff $\exists i_1, i_2, \ldots, i_j$ such that $\tau_i - \tau_{i_1} {\in} I_1 \wedge \ldots \wedge \tau_i - \tau_{i_j} {\in} I_j$ and the untimed behaviour of the propositions in $\Sigma$ is given by the MSO[$<$] formulae $\varphi(t_0=i,t_1 = i_1, \ldots, t_j=i_j)$. We add extra monadic predicates from $\I \cup \{\anch\}$ to get an interval word encoding the timed behaviour of $\varphi$.
By definition of ``consistency relation'' for interval words, any pointed timed word $\rho,i {\models} \psi(t_0)$ iff there exists an $\I = \{I_1, \ldots, I_j\}$ interval word $w$ such that $w {\models} \psi_{ut}$ and $\rho,i{\in} Time(w)$, where $\psi_{ut} {=} \exists t_0. [\anch(t_0) \wedge \exists t_1 I_1(t_1) \ldots \exists t_j I_j(t_j). \varphi(t_0, \ldots, t_j) \wedge \forall t. (\anch(t) \rightarrow (t=t_0))]$. Hence, $\rho,i {\in} Time(L(\psi_{ut}))$ iff $\rho,i {\models} \psi(t_0)$. Note that $\psi_{ut}$ only accepts valid interval words.
Recall that for any $\I$ interval word $w$, for any point $i' {\in} dom(w)$, the truth of the predicate $I(i')$ implies that $\tau_{i'} - \tau_i {\in} I$ for some $I_x {\in} \I$. But $\tau_{i'} - \tau_i {\in} I$ doesn't necessarily imply that $I(i')$ is true. Hence, these monadic predicates only witness positive satisfaction of the timing condition. This is the key reason why we had to get rid of the universal metric quantifier.
For example, if $\psi(t_0) = \tforall t {\in} t_0 + I.~a(t)$, then corresponding $\psi_{ut} = \tforall t {\in} t_0 + I.~\varphi(t_0,t)$. But $\psi_{ut}$ vacuously accepts all the interval words where $I$ doesn't appear at all. Hence, $Time(\psi_{ut})$ accepts all the pointed timed words and fails to encode the language accepted by $\psi$ as $\psi$ only accepts words $\rho,i$ such that $a$ holds at all the points within interval $I$ from $i$.
\\2) By B\"{u}chi Elgot Trakhtenbrot Theorem \cite{buchi} \cite{Elgot} \cite{Trakhtenbrot}, a $\text MSO[<]$ sentence $\psi_{ut}$  can be reduced to an equivalent NFA $A = (2^{\Sigma\cup \I \cup \{\anch\}},Q, init, F, \delta)$ over $\I$ interval words.
By lemma \ref{lem:nfatopnmtl}(lemma 4 \cite{KKMP21}), for any NFA over $\I$ interval words we can construct a PnEMTL formulae $\phi$ such that for any pointed timed word $\rho,i$, $\rho,i {\in} Time (L(A))$ iff $\rho,i {\models} \phi$. Hence,  $\rho,i {\models} \psi(t_0)$ iff there exists $\I$ interval word $w {{\in}} L(A)$ such that $\rho,i {{\in}} Time(w)$ iff 
$\rho,i {{\in}} Time(L(A))$ iff $\rho, i {\models} \phi$. Moreover, if $\psi$ is non-adjacent then, $\I$ is non-adjacent and thus $\varphi$ is in NA-PnEMTL.
\\Assume that the lemma holds for all formula of depth less than $n$.
Let $\psi(t_0)$ be any AF-GQMSO formula of metric depth $n$. 
With every timed subformulae $\psi_i(t)$ of $\psi$, we associate a witness proposition $b_i$ such that $b_i$ holds iff $ \psi_i$ holds. Let $W$ be the set of witnesses.
We replace each subformula $\psi_i(t)$ of type $MSO^T$ with its corresponding witness getting a formula $\psi'(t_0)$. As $\psi'(t_0)$ doesn't contain any subformulae of the form $MSO^T$, its metric depth is 1.  As shown in the base case, we can construct a PnEMTL formula $\varphi'$ equivalent to $\psi'(t_0)$ containing symbols from $\Sigma \cup W$.
Note that all subformulae $\psi_i(t_0)$ of $\psi$ are of metric depth less than $n$. Hence, by the induction hypothesis, we can construct a PnEMTL formula $\varphi_i$ equivalent to $\psi_i(t_0)$. Hence, the witnesses for $\psi_i$ are also that for $\varphi_i$. Replacing the witnesses $b_i$ with its corresponding PnEMTL formulae $\varphi_i$, we get the required PnEMTL formulae $\varphi$. Also note that if $\psi$ is non-adjacent then all its subformulae $\psi_i$ and formula $\psi'$ are non-adjacent too. This implies that formulae $\varphi_i$, $\varphi'$ and, hence $\varphi$ are NA-PnEMTL formulae. Appendix \ref{app:gqmsored} gives an example of a GQMSO formula with its equivalent PnEMTL formula.
\end{proof}

%% file: discussion.tex
\textbf{\large{6~~~~~Conclusion and Discussion}}
\\
\\\textbf{Conclusion}: We established the expressiveness equivalences between timed logics and automata as given in Equation (\ref{eqn:main1}) in the introduction. Thus, we have extended the results of \cite{KKP18} to logics and automata with both future and past. Doing this requires new techniques of abstracting timed words by symbolic anchored interval words, and leveraging the results on untimed logics and automata. Moreover, We have applied the newly proposed non-adjacency restriction from \cite{KKMP21} to the three formalisms of Equation (\ref{eqn:main1}) and shown that this makes them all decidable. The fact that the alternation of metric quantifiers in an anchored block can be eliminated using non-metric quantifiers in GQMSO (see Theorem \ref{lem:afqkmsoqkmso}) is an interesting result.
On careful inspection, it is evident that GQMLO (first order fragment of GQMSO) is equivalent to Partially Ordered 2-Way 1-ATA (PO 2-way 1-ATA).
\\\textbf{Discussion}: All our results, including decidability, extend to infinite timed words  by a suitable adaptation of our formalisms (i.e., B\"{u}chi acceptance condition for 2-Way 1-ATA, and allowing B\"{u}chi Automata modalities for PnEMTL).
Finally we pose the following open questions raised by the results introduced in this paper.
\\1) Unlike 1-ATA, the 2-Way 1-ATA are able to express the language $\mathsf{L_{insterr}}$ (see Example \ref{ex:insterr}). This poses a very natural question: what subclass of timed regular languages can be accepted by 2-Way 1-ATA? Does the clock hierarchy collapses at 1-clock by adding 2-Wayness?
\\2) Non-punctual Q2MLO \cite{rabin}, the most  expressive known decidable fragment of FO[$<,+n$], is a syntactic subclass of GQMLO. Is NA-GQMLO strictly more expressive than non-punctual Q2MLO? A positive answer would make NA-GQMLO the most expressive decidable fragment of FO[$<,+n$]. A negative answer would imply that non-punctual Q2MLO is equivalent to NA-PO 2-way 1-ATA.
\\3) Is non-adjacent PnEMTL strictly more expressive than EMITL \cite{Wilke}? A negative answer implies a tight automata and MSO logic characterization of EMITL.

%


%% file: appendix.tex
\section{Useful Notations for the rest of the Appendix}
We give some useful notations that will be used repeatedly in the following proofs.
\begin{enumerate}
    \item For any set $S$ containing propositions or formulae, let $\bigvee S$ denote $\bigvee \limits_{s\in S}(s)$. Similarly, let $A = \{I_1, \ldots, I_n\}$ be any set of intervals. $\bigcap A = I_1 \cap \ldots \cap I_n, \bigcup A = I_1 \cup \ldots \cup I_n$. For any automaton $A$ let $L(A)$ denote the language of $A$.
 \item For any NFA $A = (Q,\Sigma, i,F,\delta)$, for any $q \in Q$ and 
 $F' \subseteq Q$. $A[q,F'] = (Q,\Sigma, q, F', \delta)$. In other words, $A[q,F']$ is the automaton where the set of states and transition relation are identical to $A$, but the initial state is $q$ and the set of final stats is $F'$. 
 For the sake of brevity, we denote $A[q,\{q'\}]$ as $A[q,q']$. Let $\rev(A) = (Q\cup\{f\},\Sigma, f,\{i\},\delta')$, where $\delta'(f,\epsilon) = F$, for any $a \in \Sigma,q \in Q$, $(q,a,q') \in \delta'$ iff $(q',a,q) \in \delta$. In other words, $\rev(A)$ is an automaton that accepts the reverse of the words accepted by $A$.
\item Given any sequence Str, let $|$Str$|$ denote length of the sequence Str. Str[$x$] denotes $x^{th}$ letter of the sequence if $x\le|$Str$|$. Str[$1...x$] denotes the prefix of the string Str ending at position $x$. Similarly, Str[$x...$] denotes the suffix of the string starting from $x$ position. Let $S_1, \ldots S_k$ be sets. Then, for any $t \in S_1 \times \ldots \times S_k$ if $t = (x_1, x_2, \ldots, x_k)$. $t(j)$, for any $j < k$, denotes $x_j$. 
\item For a timed word $\rho$, $\rho[i](1)$ gives the set of propositions true at point $i$.  $\rho[i](2)$ gives the timestamp of the point $i$.  
\end{enumerate}
\section{Interval Word Abstraction}
\label{app:interval word}
Let $I_{\nu}\subseteq \intinterval$. 
An $I_{\nu}$-interval word over $\Sigma$ is a word $\kappa$ of the form 
$a_1 a_2  \dots a_n \in 
(2^{\Sigma \cup \{\anch\} \cup  I_\nu})^*$. There is a unique $i \in dom(\kappa)$ called the \emph{anchor} 
of $\kappa$. At the anchor position $i$,   $a_i \subseteq \Sigma \cup \{\anch\}$, 
and $\anch \in a_i$.  
Let $J$ be any interval in $\I_\nu$. We say that a point $i \in dom(\kappa)$ is a $J$-time restricted point if and only if, $J \in a_i$. 
$i$ is called time restricted point if and only if either $i$ is $J$-time restricted for some interval $J$ in $I_\nu$ or $\anch \in a_i$. 
\smallskip

\noindent \textbf{From $I_\nu$-interval word to Timed Words} : 
Given a $I_\nu$-interval word $\kappa=a_1 \dots a_n$ over $\Sigma$ and 
a timed word $\rho=(b_1, \tau_1)\dots (b_m, \tau_m)$, 
the pointed timed word $\rho, i=(b_i, \tau_i) \dots, (b_m, \tau_m)$ is consistent with $\kappa$ iff 
$dom(\rho){=}dom(\kappa)$, $i{=}\anch(\kappa)$, for all $j\in dom(\kappa)$, $b_j=a_j\cap \Sigma$ and for $j \ne i$, $I \in a_j \cap I_\nu$ implies 
$\tau_j - \tau_i \in I$. 
Intuitively, each point $j$ of $\kappa$ does the following. (i) It stores the set of propositions that are true at point $j$ of $\rho$ and (ii) It  also stores the set of intervals $\I \subseteq I_\nu$ such that the time difference between point $i$ and $j$ of $\rho$ lies within $\bigcap \I$, thus abstracting the time differences from the anchor point($i$) using some set of intervals in $I_\nu$.
We denote the set of all the pointed timed words consistent with a given interval word $\kappa$ as $\mathsf{Time(\kappa)}$. Similarly, given a set $\Omega$ of $I_\nu$ interval words, $\mathsf{Time(\Omega)}{=}\bigcup \limits_{\kappa \in \Omega} (\mathsf{Time(\kappa)})$. 
\\\noindent {\it Example.}
Let 
$\kappa{=}\small{\{a,b, (-1,0)\} \{b, (-1,0)\} 
\{a, \anch\} \{b,[2, 3]\}}$ be an interval word over the set of intervals 
$\{(-1,0),[2,3]\}$. 
Consider timed words $\rho$ and $\rho'$ s.t.
\\ $\rho{=}\small{(\{a,b\}, 0)(\{b\}, 0.5), (\{a\}, 0.95) (\{b\}, 3)}$,$\rho'{=}\small{(\{a,b\},0)(\{b\},0.8)(\{a\},0.9)(\{b\},2.9)}$.
Then $\rho,3$ as well as $\rho',3$ are  consistent with $\kappa$ while $\rho,2$ is not.  Likewise, for the timed word $\rho'' {=} (\{a,b\}, 0), (\{b\}, 0.5), (\{a\}, 1.1) (\{b\}, 3)$,  $\rho'', 3$ is not consistent with $\kappa$ as $\tau_1-\tau_3 \notin (-1,0)$, as also $\tau_3-\tau_2 \notin [2,3]$.

Note that the ``consistency relation'' is a many-to-many relation. For the set of intervals $I_\nu$, a pointed timed word $\rho,i$ can be consistent with more than one $I_\nu$-interval word and vice versa.
Let $I_\nu, I_\nu' \subseteq \intinterval$. Let $\kappa=a_1\dots a_n$ and $\kappa'=b_1 \dots b_m$ be $I_\nu$ and $I_\nu'$ interval words, respectively. 
 $\kappa$ is \emph{similar} to $\kappa'$, denoted by $\kappa \sim \kappa'$ if and only if, (i) $dom(\kappa){=}dom(\kappa')$, (ii) for all $i \in dom(\kappa)$, $a_i \cap  \Sigma{=}b_i\cap \Sigma$, and (iii)$\anch(\kappa)=\anch(\kappa')$. 
 $\kappa$ is \emph{congruent} to $\kappa'$, denoted by $\kappa \cong \kappa'$, iff $\mathsf{Time}(\kappa){=}\mathsf{Time}(\kappa')$. In other words, $\kappa$ and $\kappa$ abstract the same set of pointed timed words. Note that $\kappa \cong \kappa'$ implies $\kappa \sim \kappa'$.
\\\noindent \textbf{Boundary Points}: For any $I \in I_\nu$, $\first(\kappa, I)$ and $\last(\kappa,I)$  respectively denote the first and last $I$-time restricted points in $\kappa$. If $\kappa$ does not contain any $I$-time restricted point, then both $\first(\kappa, I){=}\last(\kappa,I){=}\bot$. We define, $\boundaryint(\kappa){=}\{i | i \in dom(\kappa) \wedge \exists I \in I_\nu$ s.t. $(i{=}\first(\kappa,I) \vee i=\last (\kappa,I) \vee i=\anch(\kappa))\}$. 
\\\noindent \textbf{Collapsed Interval Words}. Given an $I_{\nu}$ interval word $\kappa{=}a_1 \dots a_n$, let $\Ii_j$ denote
the largest set of intervals from $I_{\nu}$ contained in $a_j$.
Let $\kappa'{=}\col(\kappa)$ be the word obtained by replacing $\Ii_j \subseteq a_j$ with $\bigcap_{I \in \Ii_j} I$ in $a_j$, for all $j{\in}dom(\kappa)$.  It is clear that 
$\mathsf{Time}(\kappa){=}\mathsf{Time}(\kappa')$. $\kappa'$ is a $\Clcap(I_\nu)$ interval word, where \\$\Clcap(I_\nu){=}\{I | I{=}\bigcap I', I' \subseteq I_\nu\}$. An interval word $\kappa$ is called \emph{collapsed} iff $\kappa{=}\col(\kappa)$. 
\begin{figure}[t]
    \centering
  \scalebox{0.6}{  \includegraphics{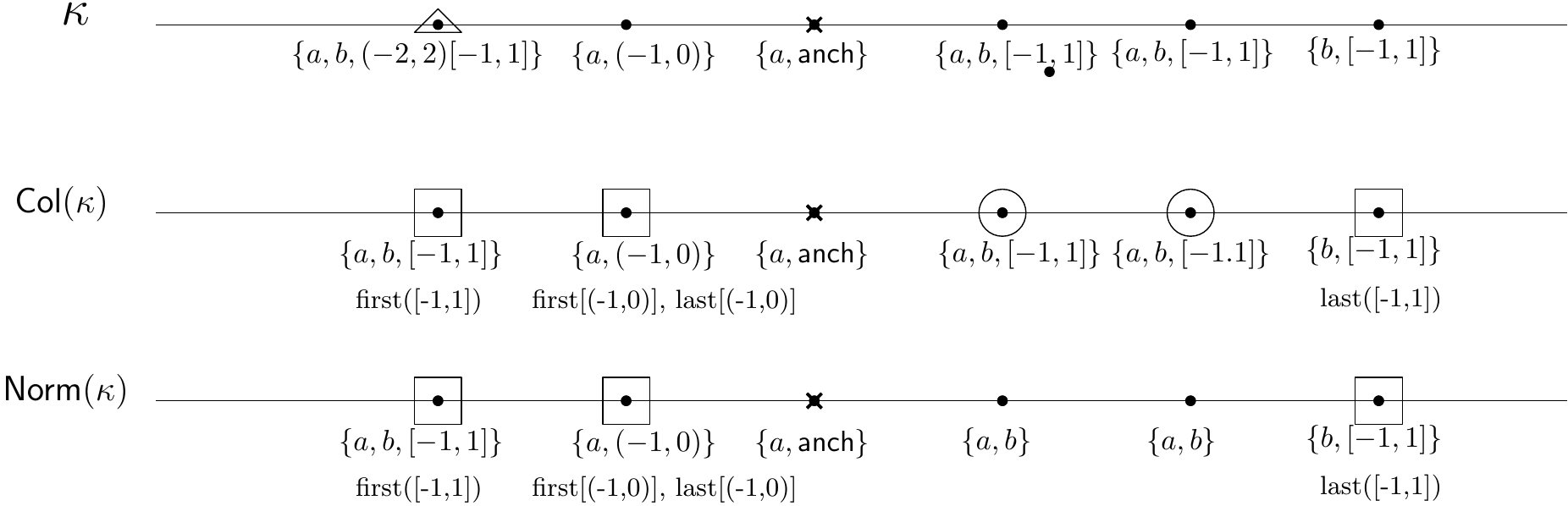}}
    \caption{The point within the triangle has more than one interval. The encircled points are intermediate points and carry redundant information. The required timing constraint is encoded by first and last time restricted points of all the intervals (within boxes). }
    \label{fig:norm_example}
\end{figure}

\noindent \textbf{Normalization of Interval Words}. Given an $I_{\nu}$ interval word $\kappa=a_1 \dots a_n$, 
we define the normalized word corresponding to $\kappa$, denoted $\normalize(\kappa)$ as a $\Clcap(I_\nu)$ interval word $\kappa_{nor}=b_1 \dots b_m, $ such that (i) 
$\kappa_{nor} \sim \col(\kappa)$, (ii) for all $I \in \Clcap(I_\nu)$, $\first(\kappa, I){=}\first(\kappa_{nor}, I)$,  $\last(\kappa, I){=}\last (\kappa_{nor}, I)$, and for all points $j \in dom(\kappa_{nor})$ s.t. $\first(\kappa, I) < j < \last(\kappa,I)$, 
$j$ is not a $I$-time constrained point. Thus, $\normalize(\kappa)$ is an $I_\nu$ interval word  \emph{similar} to $\kappa$, has identical first and last $I$-time restricted points and has no intermediate $I$-time restricted points, for any $I \in I_\nu$.  An $I_\nu$ interval word $\kappa$ is normalized iff $\normalize(\kappa){=}\kappa$. Hence, a normalized word is a collapsed word where for any $J\in I_\nu$ there are at most 2 $J$-time restricted points. 
Refer Figure \ref{fig:norm_example} for example.
\begin{lemma}
\label{lem:normalize}
$\kappa \cong \normalize(\kappa)$.Hence, any $I_\nu$ interval word, $\kappa$, can be reduced to a congruent word $\kappa'$ such that $\kappa'$ has at most $2 \times |I_\nu|^2{+}1$ time restricted points.
\end{lemma}
We split the proof of Lemma \ref{lem:normalize} into two parts. First, Lemma \ref{c1} shows $\kappa \cong \col(\kappa)$. Lemma \ref{fl2} implies that $\col(\kappa) \cong \normalize(\kappa)$. Hence, both Lemma \ref{c1}, \ref{fl2} together imply Lemma \ref{lem:normalize}.
\begin{lemma}
\label{c1}
Let $\kappa$ be a $I_\nu$ interval word and. Then $\kappa \cong \col(\kappa)$.
\end{lemma}
\begin{proof}
A pointed word $\rho,i$ is consistent with $\kappa$ iff
\begin{itemize}
    \item[(i)] $dom(\rho){=}dom(\kappa)$, \item[(ii)]$i{=}\anch(\kappa)$, \item[(iii)] for all $j\in dom(\kappa)$, $\rho[j](1)=\kappa[j]\cap \Sigma$ and \item[(iv)] for all $j \ne i$, $I \in a_j \cap I_\nu$ implies $\rho[j](2) - \rho[i](2)\in I$.
    \item[(v)]$\kappa \sim \col(\kappa)$, by definition of $\col$.
\end{itemize}
 Hence given (v),
(i) iff (a) (ii)iff (b)(iii) iff (c) where:\\
(a) $dom(\rho){=}dom(\kappa){=}dom(\col(\kappa))$, (b) $i{=}\anch(\kappa) = \anch(\col(\kappa))$, (c) for all $j\in dom(\kappa)$, $\rho[j](1)=\kappa[j]\cap \Sigma = \col(\kappa)[j] \cap \Sigma$.
(iv) is equivalent to $\rho[j](2) - \rho[i](2) \in \bigcap(\kappa[j] \cap I_\nu)$, but $\bigcap(\kappa[j] \cap I_\nu) = \col(\kappa)[j]$. Hence, (iv) iff (d) $\rho[j](2) - \rho[i](2) \in \col(\kappa)[j]$.
Hence, (i)(ii)(iii) and (iv) iff (a)(b)(c) and (d). Hence, $\rho,i$ is consistent with $\kappa$ iff it is consistent with $\col(\kappa)$.
\end{proof}
\begin{lemma}
\label{fl2}
Let $\kappa$ and $\kappa'$ be $I_\nu$ interval words such that $\kappa \sim \kappa'$. If for all $I \in I_\nu$, $\first(\kappa, I) = \first(\kappa', I) $ and $\last(\kappa, I) = \last(\kappa', I)$, then $\kappa \cong \kappa'$.
\end{lemma}
\begin{proof}
The proof idea is the following:
\begin{itemize}
    \item As $\kappa \sim \kappa'$, the set of timed words consistent with any of them will have identical untimed behaviour.
    \item As for the timed part, the intermediate $I$-time restricted points ($I$-time restricted points other than the first and the last) do not offer any extra information regarding the timing behaviour. In other words, the restriction from the first and last $I$ restricted points will imply the restrictions offered by intermediate $I$ restricted points.
\end{itemize}
Let $\rho= (a_1, \tau_1),\ldots (a_n,\tau_n)$ be any timed word.
$\rho, i$ is consistent with $\kappa$ iff
\begin{enumerate}
    \item 
    \begin{itemize}
        \item[(i)]$dom(\rho) = dom(\kappa)$, \item[(ii)] $i = \anch(\kappa)$, \item[(iii)] for all $j \in dom(\rho)$, $\kappa[j]\cap \Sigma = a_j$ and \item[(iv)]for all $j\ne i \in dom(\rho)$, $\tau_j -\tau_i  \in \bigcap (I_\nu \cap \kappa[j])$.
    \end{itemize}
    \smallskip 
    
Similarly, $\rho, i$ is consistent with $\kappa'$ if and only if 
\item 
\begin{itemize}
    \item[(a)] $dom(\rho) = dom(\kappa')$, \item[(b)]$i = \anch(\kappa')$, \item[(c)]for all $j \in dom(\rho)$, if $\kappa'[j]  \cap \Sigma = a_j$ and \item[(d)] for all $j\ne i \in dom(\rho)$, $\tau_j -\tau_i  \in \bigcap (I_\nu \cap \kappa'[j])$.

\end{itemize}

\end{enumerate}

Note that as $\kappa \sim \kappa'$, we have, $dom(\kappa) = dom(\kappa')$, $\anch(\kappa) = \anch(\kappa')$, for all $j \in dom (\kappa)$, $\kappa[j]\cap \Sigma = \kappa'[j]\cap \Sigma$. Thus, 2(a) $\equiv$ 1(i), 2(b) $\equiv$ 1(ii) and 2(c) $\equiv$ 1(iii).

Suppose there exists a $\rho, i$ consistent with $\kappa$ but there exists $j' \ne i \in dom(\rho)$, $\tau_j' -\tau_i  \notin I'$ for some $I' \in \kappa[j']$. By definition, $\first(\kappa',I') \le j' \le \last(\kappa',I')$. But $\first(\kappa',I') = \first(\kappa,I')$, $\last(\kappa',I') = \last(\kappa,I')$. Hence, $\first(\kappa,I') \le j' \le \last(\kappa, I')$. As the time stamps of the timed word increases monotonically, $x \le y \le z$ implies that $\tau_x \le \tau_y \le \tau_z$ which implies that $\tau_x - \tau_i \le \tau_y - \tau_i \le \tau_z -\tau_i$. Hence,  $\tau_{\first(\kappa,I')} - \tau_i \le \tau_{j'} - \tau_i \le \tau_{\last(\kappa,I')}- \tau_i$. But $\tau_{\first(\kappa,I')} - \tau_i \in I'$ and $\tau_{\last(\kappa,I')}- \tau_i \in I'$ because $\rho$ is consistent with $\kappa$. This implies, that $\tau_{j'} - \tau_i \in I'$ (as $I'$ is a convex set) which is a contradiction. Hence, if $\rho,i$ is consistent with $\kappa$ then it is consistent with $\kappa'$ too. By symmetry, if $\rho, i$ is consistent with $\kappa'$, it is also consistent with $\kappa$. Hence $\kappa \cong \kappa'$.
\end{proof}


\section{Non Adjacent GQMSO - Example}
\label{app:inster}
\begin{example} 
We give a NA-GQMSO formulae equivalent to $\mathsf{L_{inst err}}$ given Example \ref{ex:insterr} as follows:

$\begin{matrix} 
\exists t_0.&\begin{bmatrix}
First(t_0) \wedge\\ 
( \texists t_1 \in t_0 + (0,1).\texists t_2 \in t_0 + (0,1). ISx(t_1) \wedge ISy(t_2) \wedge\\ 
(\texists t_3 \in t_0 + (1,2). (\exists t. S(t, t_0) \wedge ISy(t))\\ 
\end{bmatrix} \wedge\\
\exists t_4. \exists t_5.&\begin{bmatrix} 
\{S(t_5, t_4) \wedge \texists t \in t_4 - (0,1). ISx(t)\} \wedge \\
\{\texists t' \in t_5 - (0,1). ISy(t')\}
\wedge \{\texists t \in t_5 - (1, \infty). ISx(t)\} \\
\wedge \{\exists t_6. S(t_6, t_5) \rightarrow (\texists t' \in t_6-(1,\infty). ISy(t')\}
\end{bmatrix}
\end{matrix}$\\

where $ISx()$ and $ISy()$ are defined in example \ref{ex:insterr} and $First(t_0)$ is defined in the same example as $\psi_1(t_0)$.

\end{example}

\section{Proofs for section \ref{sec:toatapnemtl}}
\subsection{NA-2-Way 1-ATA-rfl Accepting $\mathsf{L_{insterr}}$}
\label{app:na2wayvsata}
Theorem 2.8 \cite{LasotaW08} shows that language $\mathsf{L_{inst err}}$ presented in example \ref{ex:insterr} is not recognizable by 1-ATA. Hence, it suffices to show that the same can be expressed by NA 2-Way 1-ATA-rfl $A$ with islands $Q_0, \ldots, Q_5$ with header states $q_0, q_1, q^-_2, q^-_3,q^-_4, q^-_5$, respectively, and transition function $\delta$ detailed as follows. Backward moving locations are superscripted with $-$ sign in the following. Let $\rho = (b, 0) (b, \tau_2) (b,\tau_3) \ldots (b, \tau_n)$ be any timed word.
\begin{itemize}
    \item $Q_5 = \{q^-_{5}, q^-_{5,2}, q^-_{5,3}, q^-_{5,4}\}$, $\delta(q^-_{5}, b) = q^-_{5}$, $\delta(q^-_{5}, b, x \in (1,\infty)) = q^-_{5,2}$, $\delta(q^-_{5,2}, b) = q^-_{5,3}$, $\delta(q^-_{5,3}, b) = q^-_{5,4}$, $\delta(q^-_{5,4}, \lem) = \top$. When called from any point $i$ of $\rho$, this island makes sure that $\tau_i > \tau_3+1$.
    \item $Q_4 = \{q^-_{4}, q^-_{4,2}, q^-_{4,3}\}$, $\delta(q^-_{4}, b, x \in [0,1]) = q^-_{4} \vee q^-{4,2}$,
    $\delta(q^-_{4,2}, b) = q^-_{4,3}$,
    $\delta(q^-_{4,3}, \lem) = \top$. When called from any point $i$ of $\rho$, this island makes sure that $\tau_i \le  \tau_2+1$.
    \item $Q_3 = \{q^-_{3}, q^-_{3,2}, q^-_{3,3}\}$, $\delta(q^-_{3}, b) = q^-_{3}$, $\delta(q^-_{3}, b, x \in (1,\infty)) = q^-_{3,2}$,
    $\delta(q^-_{3,2}, b, = q^-_{3,3}$,
    $\delta(q^-_{3,3}, \lem) = \top$. When this island is called from any point $i$ of $\rho$, it makes sure that $\tau_i >  \tau_2+1$.
    \item $Q_2 = \{q^-_{2}, q^-_{2,2}, q^-_{2,3}, q^-_{2,4}\}$, $\delta(q^-_{2}, b, x\in [0,1]) = q^-_{2} \vee q^-_{2,2}$, $\delta(q^-_{2,2}, b) = q^-_{2,3}$, $\delta(q^-_{2,3}, b) = q^-_{2,4}$
    $\delta(q^-_{2,4}, \lem) = \top$. When called from any point $i$ of $\rho$, this island makes sure that $\tau_i \le \tau_3+1$.
    \item $Q_1 = \{q_1, q_{1,2}, q_{1,3}\}$, $\delta(q_1, b) = q_{1,2}$,  $\delta(q_{1,2}, b) = q_{1,3}$, $\delta(q_{1,3}, b, x\in (1,2)) = \top$. If this island is called after reading the the first position, then it makes sure that $\tau_4 \in (1,2)$.
    \item $Q_0 = \{q_0, q_{0,2}, q_{0,3},q_{0,4}, q_{0,5}, q^-_{0,6}\}$. $q_0$ is the initial location of the automata $A$.
    $\delta(q_0, b) = q_{0,2} \wedge x.q_1$. On reading the first symbol, this transition moves to a location $q_{0,2}$ and simultaneously calls island $Q_1$ from the position 1 (and timestamp 0). Moreover from $q_0$, there are two consecutive $b$ within interval $(0,1)$. This is expressed by transitions,
    $\delta(q_{0,2}, b, x \in (0,1)) = q_{0,3}$, $\delta(q_{0,3}, b, x \in (0,1)) = q_{0,4}$.
    Hence, $\rho$ is accepted by $A$ only if it has at least 4 points where $b$ holds. Moreover, the second and third points are within interval $(0,1)$ and the fourth point is within interval $(1,2)$(the latter is expressed by $Q_1$). $\delta(q_{0,4}, b) = q_{0,4}  \vee (x.q^-_2 \wedge x.q^{-}_{3} \wedge q_{0,5} \wedge q^-_{0,6})$, $\delta(q^-_{0,6}, b) = x.q^-_4$, $\delta(q_{0,5}, b) = x.q^-5$, $\delta(q_{0,5}, \rem) = \top$. Location $q_{0,4}$ loops on $b$ and non-deterministically chooses a position $i$ of $\rho$ from where it calls islands $Q_2$ and $Q_3$ simultaneously. Moreover, island $Q_4$ is called from the position $i-1$ and $Q_5$ from $i+1$. This implies that $\rho$ is accepted by $A$ iff $\tau_2, \tau_3 \in (0,1)$ and there exists exactly one point $i\in dom(\rho)$ such that $\tau_i  \in [\tau_2 + 1, \tau_3 +1]$ (as $\tau_{i-1}  \in [0, \tau_2 + 1)$, $\tau_i+1  \in (\tau_3 +1, \infty)$ or $\tau_i$ is the last position of $\rho$).
\end{itemize}

\subsection{Proof of Lemma \ref{lem:pnemtltoata}}
\label{app:pnemtltoata}
We apply induction on modal depth of the formulae $\varphi$. For modal depth 0, $\varphi$ is a propositional formulae. Hence, the lemma trivially holds. For modal depth 1, let $\varphi$ be of the form $\varphi = \fregk_{I_1, \ldots, I_k} (\re_1, \ldots, \re_{k+1})(\Sigma)$. In the case of $\sregk$ modality, symmetrical construction applies. Moreover, dealing with boolean operators is trivial as the resulting 2-Way 1-ATA-rfl are closed under boolean operations. 
Let $\re_j = (2^{\Sigma}, Q_j, init_j, F_j, \delta_j)$. 
For $a \in \Sigma$ let Pre$(a,F_j) = \{q| q\in Q_j \wedge F_j\cap\delta_j(q,a)\ne \emptyset\}$. Hence Pre$(a,F_j)$ denote set of all the locations in $\re_j$ that has a transition to an accepting state on reading $a$. By semantics, for any timed word $\rho = (a_1, \tau_1) \ldots (a_m,\tau_m)$ and $i_0 \in dom(\rho')$, $\rho,i_0 \models \varphi$ iff there exists a sequence of point $i_1, \ldots, i_k, i_{k+1}$ in strict future of $i_0$ where $i_{k+1} = m$ such that the behaviour of propositions in $\Sigma$ between the segment from $i_j+1$ to $i_{j+1}$ is given by automata $\re_{j+1}$ for any $0\le j \le k+1$. This specification can be expressed using 1-clock Non Deterministic Timed Automata, $\A = (2^{\Sigma}, Q, init, F, \delta, \G)$, constructed as follows. $Q = Q_1\ldots Q_{k+1}$. $init = init_1$, $F = F_{k+1}$, $\G = \{x \in I_1 \ldots, x\in I_k\}$, for any $1 \le j \le k+1$ $q \in Q_j$, $\delta(q,a) = \delta_j(q,a)$, for any $1\le j \le k$ $q \in \text{Pre}(a,F_i)$, $\delta(q,a,x\in I_{i}) = init_{j+1}$.
By semantics of NTA $\rho, i \models \varphi'$ if and only if, $\A$ reaches accepting state on reading $\rho$ starting from position $i$. Note that the NTA we constructed is a reset free NTA. In case, of $\sreg$ modality we will have a backward moving reset free NTA. Hence, the NTA constructed are ``2-Way 1-ATA-rfl'' with single island.

Let us assume that the lemma holds for every PnEMTL formulae of depth less than $n$. Let $\varphi$ be any PnEMTL formulae of modal depth $n$ of the form $\fregk_{I_1, \ldots, I_k} (\re_1, \ldots, \re_{k+1})(S)$.
As the $\md(\varphi) = n$,  any formulae $\phi_j \in S$ is s.t. $\md(\phi_i)< n$. We consider the set of timed words $T$ over extended set of propositions $\Sigma \cup W$, where $W$ is a set of propositions containing witness $b_j$ for each formulae $\phi_j \in S$ such that for any $\rho'\in T$ and $i \in dom(\rho')$, $\rho',j \models \phi_i \iff \rho', j \models b_i$.
Let $\varphi'$ be a formulae obtained from $\varphi$ by replacing occurrence of every $\phi_i \in S$ by its corresponding witness $b_i$. 
Given any word $\rho$ over $\Sigma$, 
Let $\rho' \Downarrow \Sigma$ denote a word $\rho$ over $\Sigma$ obtained from $\rho'$ by hiding symbols from $W$. 
For any $i \in dom (\rho')$, $\rho', i \in T$,  $\rho',i \models \varphi' \iff \rho'\Downarrow \Sigma,i \models \varphi$. Hence, any pointed word $\rho,i$ satisfies $\varphi$ if and only if it is a projection on $\Sigma$ of a timed word $\rho' \in T$ and $\rho',i \models \varphi'$.
Note that $\varphi'$ is a modal depth 1 formulae of the form $\fregk_{I_1, \ldots, I_k} (\re'_1, \ldots, \re'_{k+1})(\Sigma \cup W)$. Hence, we can construct a 2-Way 1-ATA-rfl $\A'$ with only 1 island over $\Sigma \cup W$.  

To get an automata $\A$ equivalent to $\varphi$, we need to make sure that the it accepts all and only those words $\rho, i$ where $\rho$ is timed word over $\Sigma$ which can be obtained from a word $\rho'$ over $\Sigma \cup W$ in $T$ such that $\rho',i$ is accepted by $\A$.
This can be done by as follows. By induction hypothesis, for any subformulae $\phi_i \in S$, we can construct a 2-Way 1-ATA-rfl $\A_i = (2^\Sigma, Q^+_i, Q^-_i,q_i,F_i, \delta_i,\G_i)$ and $\Ac_i = (2^\Sigma, Q'^+_i,Q'^-, q'_i,F'_i, \delta'_i,\G'_i)$ such that $\A_i$ is equivalent to $\phi_i$ and $\Ac_i$ to $\neg \phi_i$. Let $S = \{\phi_1,\ldots, \phi_n\}$, $\mathcal{Q}^{\sim} = Q^{\sim}\cup Q^{\sim}_1\cup \ldots\cup Q^{\sim}_n \cup   Q'^{\sim}_1 \cup\ldots \cup Q'^{\sim}_n$ for $\sim \in \{{+,-}\}$, $\mathcal{F} = F\cup F_1\cup \ldots\cup F_n \cup   F'_1 \cup\ldots \cup F'_n$, $\G = \G'\cup \G_1\cup\ldots \G_n \cup \ldots \cup\G'_1 \cup\ldots \G'_n$. We now construct the required automata $\A$. Intuitively, every transition (not) labelled by $b_i$ is conjuncted with a new transition to ($q'_j$) $q_j$, respectively. 
$\A = (2^\Sigma, \mathcal{Q}^+, \mathcal{Q}^-, init, \mathcal{F}, \delta, \G)$ where for any $a \in \Sigma$, $g \in \G'$, if $q \in Q$ then
\\$\delta(q,a,g) = \bigvee \limits_{W' \subseteq W} [\delta'(q,a\cup W',g) \wedge \bigwedge \limits_{b_i \in W'} x.q_i \wedge \bigwedge \limits_{b_i \notin W'} x.q'_i]$, \\if $q \in Q_i$ then
$\delta(q,a,g) = \delta_i(q,a,g)$, if $q \in Q'_i$ then $\delta(q,a,g) = \delta'_i(q,a,g)$. 

Note that, by construction, each of $Q, Q_1, \ldots, Q_n, Q'_1,\ldots, Q'_n$ for islands of $\A$. Moreover, if $\varphi$ is non adjacent then island $Q$ uses non adjacent set of intervals as all its outgoing transitions use the same set of intervals as used by $\Ac$. Also if $\varphi$ is non adjacent then all the its subformulae in $S$ are non adjacent. By inductive hypothesis islands $Q_1, \ldots, Q_n, Q'_1,\ldots, Q'_n$ are also non adjacent. This proves the lemma.
Note that if $\varphi$ is a $\sregk$ formula then the initial island would have been a backward moving island.\subsection{Proof of Lemma \ref{lem:toatapnemtl}}
We apply induction on reset depth of $\A$. The key difference between reduction from 1-ATA-rfl to $\regmtl$ in \cite{KKP18} is in the reduction of single island (or reset free) Automata to an equivalent formulae (base case). In \cite{KKP18}, the reduction was via region abstraction of words. These abstractions do not preserve the non adjacency restriction. In this case, we use a coarser abstraction of interval words which helps us to preserve non-adjacency while reduction and hence get decidable fragments for 2-Way 1-ATA. We give the flow of the construction here. 

We break the construction into following steps:
\begin{enumerate}
    \item \textbf{Base Case}: For reset depth 0, $\A$ is reset free 2 way 1-ATA-rfl. Reduce $\A$ to AFA $\absa$ over $\I$ intervals words such that $\rho,i \models_\A init$ if and only if $\rho,i \in Time(L(A))$.
    \item Reduce 2-Way AFA $\absa$ to 1-way NFA $A$ over $\I$ interval words using result from \cite{OG14}.
    \item Show that given any NFA $A$ over $\I$ interval words, one can construct a $\pnregmtl$ formula $\varphi$ such that $\rho,i \in Time(L(A)) \iff \rho,i \models \varphi$. Hence, $\rho, i \models \varphi \iff \rho,i \models_\A (init, 0)$. Moreover, if $\I$ is a non adjacent set of intervals then $\varphi$ is a non adjacent PnEMTL formula.
    \item \textbf{Induction}: Replace all the lower level islands by witness propositions. Then apply the reduction as in base case. Finally, the witnesses are replaced by subformula equivalent to the corresponding automata.
    
\end{enumerate}
\label{app:toatapnemtl}
\subsubsection{Example To Present the Reduction from 2-Way 1-ATA-rfl to PnEMTL}
\label{app:eg}
\begin{figure}
    \centering
    \includegraphics[scale=0.5]{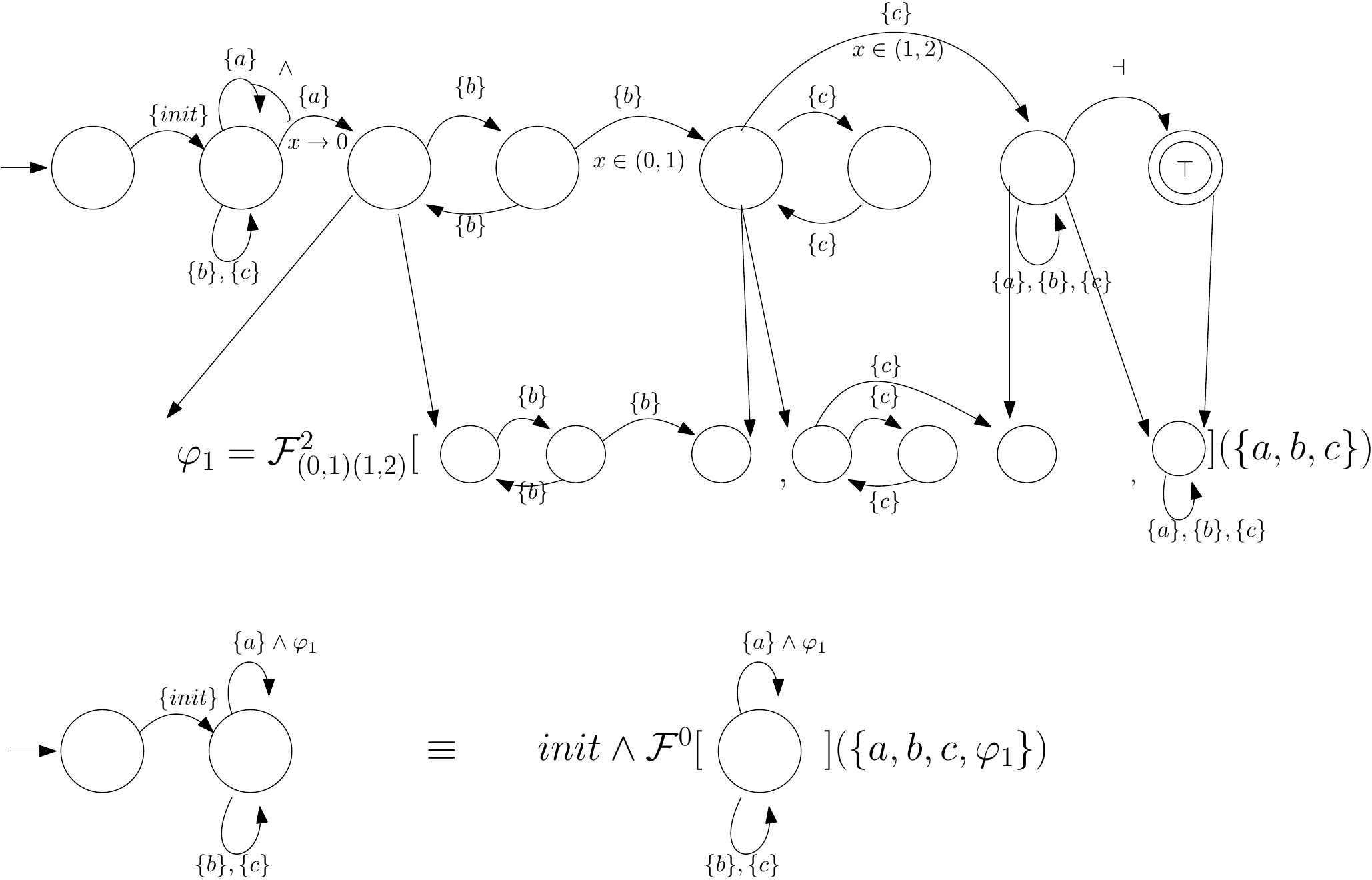}
    \caption{Example showing reduction from 2-Way 1-ATA-rfl to PnEMTL.}
    \label{fig:my_label}
\end{figure}
\subsubsection{From Reset Free 2-Way 1-ATA to AFA over interval words}
\label{app:toataafa}
We first give a construction of two way AFA over $\I$ interval words constructed from $\mathcal{A}$ as follows, denoted by $ABS(\mathcal{A})$, such that $\rho,i \in Time (L(\absa)) \iff  \rho,i \models_{\mathcal{A}} (init, 0)$. $\absa = (2^{\Sigma \cup \{\anch\} \cup I_\nu} , Q^+ \cup \{check, init'\}, Q^-, init', \top ,\bot, \delta)$ such that for any $q \in Q^+ \cup Q^-$, $a \in \Gamma \cup \{\lem, \rem\}$, $\J \subseteq \I$, $\delta(q, a\cup \J) = \bigvee \limits_{I \in \J} \Delta (q, a, x \in I)$, $\delta(q, a\cup  \{\anch\}) = \bigvee \limits_{0 \in I'} \Delta (q, a, x \in I')$.
That is, for every conjunction of outgoing edges from a location $q$ to a set of locations $Q'$ on reading $a \in \Gamma$ with guard $x \in I$  in $\mathcal{A}$, there is a conjunction of outgoing edges from state $q$ to $Q'$ on reading symbol $a \cup \J$ for any $\J \subseteq \I$ and $I \in \J$ or on reading a symbol $a \cup \anch$ if $0 \in I$. Moreover, for any $a \in \Gamma$ and $I \in I_\nu$,
\begin{itemize}
 \item $\delta(init', a\cup \J) = init'$, $\delta(init', a\cup \{anch\}) = check \wedge init$ : Continue to loop till the anchor point is encountered. After reading anchor point, spawn two locations , $check$ and $init$, simultaneously.
 \item $\delta(init', \rem) = \bot$: If no anchor point is encountered before the head reaches the right end marker, reject the word.
\item $\delta(check, a \cup \J) = check$, $\delta(check, \rem) = \top$, $\delta(check, a \cup \{anch\}) = (\bot, r)$, $\delta(init', \rem) = \bot$: Continue to loop on check after the first encounter of an anchor point. If another anchor point is encountered, reject the word.
\item The above conditions will makes sure that a word is accepted only if it has exactly one anchor point and thus is a valid $I_\nu$-interval word.\end{itemize}

\subsubsection{Proof of lemma \ref{lem:toataafa}-Converse Direction [$\Leftarrow$]} 
\label{app:toatafalem}
$\Leftarrow$ 
Let  $\rho  = (a_1,\tau_1) \ldots (a_m, \tau_m)$. Let $a_0 = \lem$, $a_{m+1} =\rem$, $\tau_0 = 0$ and $\tau_{m+1} = \tau_m$ and for any $0 \le j \le m+1$  let $T_j = \tau_j - \tau_i$.
Consider a word $w =a_1\cup \J_1 \ldots a_m\cup\J_m$ where $\J_i = \{anch\}$ and for all $j \in dom(w)$ and $j \ne i$, $\J'_j = {I | I \in \I \wedge T_j \in I}$.  Hence, $\J'_j$ is a maximal subset of $\I$ such that for all intervals $I \in \J_j$, $T_{j} \in I$. Clearly, $\rho,i \in Time(w)$. Moreover, $\delta(q, a_\h \cup \J_\h) = \tran(q, a_\h, T_\h)$ for any $0 \le h \le m+1$ and for any $q \in Q$. Let $s$ ($s'$) be any state of $\absa$ ( $\A$). Let $\{s\}$ be equivalent to $\{s'\}$.
Hence, for any configuration $C \in \succs_\delta(w,s)$ there exists a $\C \in \succs_\Delta(\rho,s')$ 
  such that $C$ is equivalent to $\C$. Moreover, for any configuration $\C \in \succs_\Delta(\rho, s')$ there exists a configuration $C' \in \succs_\delta(w,s)$ such that $C'$ is equivalent to $\C$. For any $g \ge 0$ and for any $\C \in \s_{\Delta}^{g}(\rho,\{(init, 0, i)\})$ there exists $C \in \s_{\delta}^{g}(w, \{(init,i)\})$ such that $C$ is equivalent to $\C$(obs 2).
For $g = 0$ the above statement is trivially true. Assume for $g = k$ the statement is true. Let $\C' = \{(q_1,\h_1, T_{\h_1}), \ldots, (q_n,\h_n, T_{\h_n})\}$ of $\A$ such that $\C' \in \s_{\Delta}^{k}(\rho, \{(init,0,i)\}\}$. 
Then by induction hypothesis, $C' = \{(q_1,\h_1), \ldots, (q_n,\h_n)\} \in \s_{\Delta}^{k}(w, \{(init,i)\})$. 
Any configuration $\C'' \in  \s_\Delta(\rho, \C')$ if and only if $\C'' = \C'_1 \cup \ldots \cup \C'_n$, where for any $1\le j \le n$, $\C'_j\in \succs_\A(q_i, T_{\h_i},\h_i)$. Let $C'' = C'_1\cup C'_n$ such that $C'_j$ is equivalent to $C'_j$ for all $1\le j \le n$.  By (obs 2), for any 
 $1\le j \le n$, $C'_j \in \succs_\delta(q_j,\h_j)$. As a result,  $C'' \in \s_\delta(\rho, \C')$. Hence, for any configuration $\C'' \in \s^{k+1}_{\A}(\rho,\{(init, 0, i)\})$ there exists a configuration $\C'' \in \s^{k+1}_{\Delta}(w,\{(init, i)\})$ such that $C''$ is equivalent to $\C''$. Hence, if $\rho, 0, i \models_\A init$ then there exists $w'$ such that $\rho,i \in Time(w')$ and $w', i \models_\absa init$.

\subsubsection{2-Way AFA to NFA}
\label{app:afanfa}
\begin{theorem}[\cite{OG14}]
\label{lem:afanfa}
For any 2-Way Alternating finite Automata $A$, one can construct a 1 Way Non Deterministic Finite Automata (NFA) $A'$ with at most exponential number of states.
\end{theorem}
We use above theorem to construct 1-Way NFA $A'$ equivalent to $\absa$.
\subsubsection{From NFA over Interval words to PnEMTL}
\label{app:nfapnemtl}
\begin{lemma}
\label{lem:nfatopnmtl} 
Given any NFA $A'$ over $\I$ interval words, we can construct a PnEMTL formula $\varphi$ such that $\rho,i \models \varphi \iff \rho, i \in Time(L(A'))$. Moreover, if $\I$ is non adjacent then $\varphi$ is a non adjacent PnEMTL formulae. 
\end{lemma}
We encourage readers to first read section \ref{app:interval word} to for definition and notations not introduced in the main paper but is used in this proof.
\smallskip

\textbf{Automata over Collapsed Interval Word}- From $A'$, we  construct an automaton $A{=}(Q, \init, 2^{\Sigma'}, \delta, F)$  s.t. $L(A){=}\col(L(A_\alpha))$. For any $q,q' \in Q$, $S \subseteq \Sigma'$, $\I \subseteq I_\nu$, $(q,S \cup \I,q') \in  \delta'$ iff $(q,S\cup\{I\},q') \in  \delta$ where $I{=}\bigcap \I$,   
$(q,\{a, \anch\},q') \in  \delta'$  iff $(q,\{a,\anch\},q') \in  \delta$. $A$ is obtained from $A'$ by replacing  $\bigcap \I$ in place of $\I$ on the transitions. This gives $L(A){=}\col(L(A'))$. 
\smallskip

\textbf{Partitioning Interval Words}-
We discuss here how to partition $W$, the set of all $I_\nu$ interval words  using  some finite sequences $\seq$ over $I_{\nu} \cup \{\anch\}$.
For any collapsed $w \in W$, $\seq$  gives an ordering between $\anch (w)$, $\first(w,I)$ and $\last(w,I)$ for all $I \in I_\nu$, such that, any $I \in I_{\nu}$ appears exactly twice 
and $\anch$ appears exactly once in $\seq$. For instance, $\seq=I_1 I_1 \anch  I_2 I_2$ is a sequence 
different from $\seq'=I_1I_2 \anch I_2I_1$ since the relative orderings between the first and last occurrences of $I_1, I_2$ and $\anch$ differ in both. Let  $\mathcal{T}(I_\nu)$ be the set of all such sequences;  by definition, $\mathcal{T}(I_\nu)$ is finite. 
Given $w \in W$, let $\boundaryint(w){=}\{i_1, i_2, \ldots, i_k\}$ be the 
positions of $w$ which are either $\first(w,I)$ or $\last(w,I)$ for some $I \in I_{\nu}$ 
or is $\anch(w)$. Let $w\downarrow_{\boundaryint(w)}$ be the subword of $w$ obtained by projecting $w$ to the positions in $\boundaryint(w)$, restricted to the sub alphabet $2^{I_{\nu}}\cup\{\anch\}$. For example, $w=\{a,I_1\}\{b,I_1\}\{c,I_2\}\{\anch,a\}\{b,I_1\}\{b,I_2\}\{c,I_2\}$
gives $w\downarrow_{\boundaryint(w)}$ as $I_1I_2\anch I_1I_2$. 
Then $w$ is in the partition $W_{\seq}$ iff  $w\downarrow_{\boundaryint(w)}=\seq$. Clearly, $W=\uplus_{\seq \in \mathcal{T}(I_\nu)}W_{\seq}$.
Continuing with the example above, $w$ is a collapsed $\{I_1, I_2\}$-interval word over $\{a,b,c\}$, with $\boundaryint(w)=\{1,3,4,5,7\}$, and $w\in W_{\seq}$ for $\seq=I_1I_2\anch I_1 I_2$, while 
$w\notin W_{\seq'}$ for $\seq'=I_1I_1\anch I_2 I_2$. 
Finally, all the timed words abstracted by interval words in a partition $W_{\seq}$ for $\seq{=}I_1'\ldots I_m' \anch I_1 \ldots I_n$ 
is expressed using (disjunction of) formulae of the form $\fregm^n_{I_1,{\ldots}, I_n} (\re_1,{\ldots}, \re_{n+1}){\wedge} \sregm^m_{I'_1, {\ldots}, I'_m} (\re'_1,{\ldots}, \re'_{m+1})$. 
\smallskip

\textbf{Construction of NFA for each type}-
Let $\seq$ be any sequence in $\mathcal{T}(I_\nu)$. In this section, given $A{=}(Q, \init, 2^{\Sigma'}, \delta', F)$ as constructed above, we construct an NFA 
$A_\seq{=}(Q \times \{1,2,\ldots |\seq|+1\}\cup \{\bot\}, (\init, 1), 2^{\Sigma'}, \delta_\seq, F \times \{|\seq|+1\})$ such that $L(A_\seq){=}\normalize(L(A) \cap W_\seq)$. 
Intuitively, the second element of the state indicates the next time restricted point expected to be read. 
More precisely, from any state $(q,j)$ the automaton does not have any transition on a time restricted point labelled $S$ if $S \cap \seq[j]{=}\emptyset$. Moreover, from any state $(q,j)$, on reading an unrestricted point of the form $S \subseteq \Sigma$, it non determinstically proceeds to a state $(q',j)$ if and only if, in automaton $A$, there is a transition of the form $q \stackrel{S\cup{J}}{\rightarrow}q'$ where $J{=}\emptyset$ or $J$ is any interval in $I_\nu$ such that $\first(J,w)$ has already been read and $\last (J,w)$ is yet to be read in the future.  

Let $\seq$ be any sequence in $\mathcal{T}(I_\nu)$. Given $A{=}(Q, \init, 2^{\Sigma'}, \delta', F)$ over collapsed interval words from LTL formula $\alpha$. We construct an NFA  $A_\seq{=}(Q \times \{1,2,\ldots |\seq|+1\}\cup \{\bot\}, (\init, 1), 2^{\Sigma'}, \delta_\seq, F \times \{|\seq|+1\})$ such that $L(A_\seq){=}\normalize(L(A) \cap W_\seq)$.

For any $(q,i) \in  Q \times \{1, \ldots, |\seq|+1 \}$,  $S \in 2^{\Sigma \cup I_\nu \cup \anch}$ and $I \in I_\nu \cup \{\anch\}$ such that $\seq[i]{=}I$, $\delta_\seq$ is defined as follows:
\\$\bullet$ If $1\le i \le |\seq|$ 
\begin{itemize}
    \item {\bf(i)} If $\seq[i] \in S$, then $\delta_\seq((q,i), S){=}\delta(q, S) \times \{i+1\}$
\item {\bf(ii)} If $\seq[i] \notin S \wedge S \setminus \Sigma \ne \emptyset$, then $\delta_\seq((q,i), S){=}\emptyset$ 
\item {\bf(iii)} If $S \setminus \Sigma{=}\emptyset$, then $\delta_\seq((q,i), S){=}[\bigcup \limits_{I’ \in \I_i} \delta (q,S \cup \{I'\}) \cup \delta_\seq(q, S)] \times \{i\}$ where $\I_i{=}\{I' | I' \in I_{\nu} \wedge \exists i',i''. i' < i \le i''$, $\seq[i'] {=}\seq[i'']{=}I'\}$.
\end{itemize}
$\bullet$ If $i{=}|\seq|{+}1$, $\delta_\seq((q, i), S){=}\emptyset$ if 
$S \setminus \Sigma \ne \emptyset$, 
$\delta_\seq((q,i), S){=}\delta(q, S)$ if 
$S \setminus \Sigma{=}\emptyset$

Let $W_\seq$ be all the set of $I_\nu$ intervals words over $\Sigma$ of type $\seq$.

\begin{lemma}
$L(A_\seq){=}\normalize(L(A) \cap W_\seq)$. Hence, $\bigcup \limits_{\seq \in \mathcal{T}(I_\nu)} L(A_\seq) {=}\normalize(L(A))$.
\label{lem:nfatonfaseq}
\end{lemma}

\begin{proof}
Let $w$ be any collapsed timed word of type $\seq$ and $w'{=}\normalize (w)$. Let $\bseq(w){=}\bseq(w')= i_1 i_2 \ldots i_n$ be the 
boundary positions. 
\begin{itemize}
    \item 
(i) If a state $q$ is reachable by $A$ on reading first $j$ letters of $w$, then $(q, k)$ is reachable by $A_\seq$ on reading the corresponding first $j$ letters of $w'$ where $i_{k-1} < j < i_{k}$.
\item (ii) If a state $(q,k)$ is reachable by $A_\seq$ on reading first $j$ letters of $w'$, then $q$ is reachable by $A_\seq$ on reading the corresponding first $j$ letters of $w$ and $i_{k-1} < j < i_{k}$.
\end{itemize}
The above two statements imply that on reading any word $w \in W_\seq$, $A$ reaches the final state if and only if $A'$ reaches the final state on reading $w'{=}\normalize(w)$.
Statement (i) and (ii) are formally proved in Lemma \ref{lem:AsubsetAseq} and Lemma \ref{lem:AseqsubsetA}, respectively. By (i) and (ii), we get 
$L(A_\seq) \cap W_\seq \supseteq \normalize(L(A) \cap W_\seq)$.

(1). By Proposition \ref{prop: aseqnorm} (below) $L(A_\seq) \subseteq \normalize (W_\seq) \subseteq W_\seq$, and (2). $L(A_\seq) \cap W_\seq{=}L(A_\seq)$.  Hence, by (1) and (2), $L(A_\seq) \subseteq \normalize(L(A) \cap W_\seq)$.
\end{proof}

\begin{proposition}
\label{prop: aseqnorm}
$L(A_\seq) \subseteq \normalize(W_\seq)$
\end{proposition}
\begin{proof}
Let $Q_i{=}Q{\times}\{i\}$. By construction of $A_\seq$, transition from a state in $Q_{i}$ to $Q_{i'}$, where $i{\ne}i'$ happens only on reading an interval $I{=}\seq [i]$\footnote{Let $I$ be any symbol in $I_\nu \cup \{\anch\}$. By ``reading of an interval $I$'' we mean ``reading a symbol $S$ containing interval $I$''.}. Moreover, $i'{=}i{+}1$. Thus, any word $w$ is accepted by $A_\seq$ only if
there exists $1{\le} i_1{<}i_2{<}\ldots{<}i_{|\seq|}{\le}|w|$ such that $w[i_k] \setminus \Sigma{=}\{\seq[i_k]\}$ and all other points except $\{i_1,\ldots,i_k\}$ are unrestricted points. This implies, $w{\in}L(A_\seq){\rightarrow}w{\in}\normalize(W_\seq)$.
\end{proof}

Let the set of the states reachable from initial state, $\init$, of any NFA $C$ on reading first $j$ letters of a word $w$ be denoted as $C<w,j>$. Hence, $A<w, 0> = \{\init\}$ and $A_\seq<w,0> = \{(\init, 1)\}$.
\begin{lemma}
\label{lem:AsubsetAseq}
Let $w$ be any collapsed $I_\nu$ interval word of type $\seq$ and $\bseq(w)  = i_1 i_2 \ldots i_n$. Let $w' = \normalize(w)$. Hence, $\bseq(w) = \bseq(w')$. For any $q \in Q$, $q \in A<w, j>$ implies $ (q, k) \in A_\seq<w', j>$ where $i_{k-1} < j < i_{k}$.
\end{lemma}
\begin{proof}
Recall that $\bseq$ is the sequence of boundary points in order. 
We apply induction on the number of letters read, $j$. Note that for $j = 0$, by definition, $A<w, 0> = \{\init\}$ and $A_\seq(\normalize(w), 0) = \{(\init, 1)\}$ the statement trivially holds as $0 < i_1 \ldots <i_n$. Let us assume that for some $m$, for every state $q \in A<w, m>$  there exists  $(q, k) \in A_\seq<\normalize(w), m>$ such that $i_1 <\ldots <i_{k-1} \le m < i_k <\ldots i_n$. Now let $j = m+1$. 
Let us assume that $q'$ is any state in $A<w,m+1>$. We just need to show that for some $(q',k') \in A_\seq<w', m+1>$ where $k' = k+1$ if $m+1 \in \boundaryint(w)$. Else $k' = k$. 

As $q' \in A<w, m+1>$, there exists a state $q \in A<w,m>$ such that $q' \in \delta (q, w[m+1])$. By induction hypothesis, $(q, k) \in A_\seq<\normalize(w), m>$.  Note that $(q', k') \in \delta_\seq((q,k), w'[m+1])$ implies $(q',k') \in A_\seq<w', m+1>$.
Let $w[m+1] = S_J$, where $S_J \subseteq \Sigma \cup I_\nu \cup \{\anch\}$ and $S_J \setminus \Sigma$ contains at most 1 element.
\\\textbf{Case 1}: $m+1 \in \boundaryint(w)$. This implies that $w'[m+1] = w[m+1]$. As both $w$ and $w'$ are of type $\seq$, $\{\seq[k]\} = S_J \setminus \Sigma$(by definition of $\seq$). Hence, by construction of $A_\seq$, $\delta_\seq((q,k),S_J) = \delta(q,S_J) \times \{k+1\}$. As $q' \in \delta(q, S_J)$,  $(q',k+1) \in \delta_\seq((q,k),S_J)$.
\\\textbf{Case 2}: $m+1 \notin \boundaryint(w)$. This implies that $w'[m+1] = S = S_J \cap \Sigma $. \\

\noindent \textit{Case2.1}:$S = S_J$. By construction of $A_\seq$, $\delta ((q,k), S) \supseteq \delta (q, S) \times {k}$. Thus, $(q', k) \in \delta_\seq((q,k),S_J \cap \Sigma)$.\\

\noindent \textit{Case2.2}:$S \ne S_J$. Let $S_J \setminus \Sigma = \{J\}$ where $J \in I_\nu \cup \{\anch\}$. Then $m+1$ is neither the first nor the last $J-$ time restricted point nor the anchor point in $w$. Hence, $\first(J, w) < m+1 < \last(J,w)$. By induction hypothesis, $i_{k-1}\le m < i_{k}$. Note, as $m+1$ is not in $\boundaryint(w)$, $m+1 \ne i_k$. Hence, $i_{k-1} \le m < m+1 < i_k$. This implies, $\first (J,w) < i_k \le \last(J,w)$. By definition of $\seq$, there exists $k'$ and $k''$ such that $k' < k \le k''$ and $\seq[k'] = \seq[k''] = J$. Hence, by construction of $\delta_\seq$, $\delta_\seq((q,k), S) \supseteq \delta(a,S_J) \times \{k\}$.
Hence $(q', k) \in \delta_\seq((q,k), S_J)$.
\end{proof}
\begin{lemma}
\label{lem:AseqsubsetA} 
Let $w'$ be any normalized $I_\nu$ interval word of type $\seq$ and \\$\bseq(w')  = i_1 i_2 \ldots i_n$. Let $i_0 = 0$. For any $q \in Q$, $(q,k) \in A_\seq<w', j>$ implies there exists a collapsed $I_\nu$ interval word $w$, such that $\normalize(w) = w'$, $q \in A_\seq<w, j>$  and $i_{k-1}< j \le i_k$
\end{lemma}
\begin{proof}
We apply induction on the value of $j$ as in proof of Lemma \ref{lem:AsubsetAseq}. For $j = 0$, the statement trivially holds. Assume that for $j = m$, the statement holds and $(q',k') \in A_\seq<w', m+1>$(Assumption 1). We need to show 
\begin{itemize}
    \item [(i)]$i'_{k'-1} \le m+1 <i'_{k'}$ and, \item[(ii)] there exists $w$ such that $\normalize(w) = w'$ and $q' \in A <w, m+1>$.  
$(q',k') \in A_\seq<w', m+1>$ implies, there exists $(q, k) \in A_\seq<w', m>$ such that $(q',k') \in \delta_\seq((q,k),w'[m+1])$. 
\end{itemize}
By induction hypothesis, $i_{k-1} \le m < i_{k}$ [IH1] and  there exists a word $w''$ such that $\normalize(w'') = w'$ and $q \in A<w'', m>$ [IH2].

\noindent \textbf{Case 1 } $m+1 \in \boundaryint(w')$: 
This implies  
\begin{itemize}
    \item[(a)] $m+1 \in \{i_1, i_2, \ldots, i_k\}$.
    \item[(b)] $k' = k+1$ (by construction of $\delta_\seq$).
    \item[(c)] $w''[m+1] = w'[m+1] = S \cup \{J\}$ such that $S \subseteq \Sigma $ and $J \in (I_\nu \cup \{\anch\})$.
\end{itemize}
 In other words, $m+1$ is either a time restricted point or an anchor point in both $w''$ and $w'$. d)$\seq[i_k'] = J$, otherwise $\delta_\seq((q,k), S \cup \{J\})) = \emptyset$ which contradicts Assumption 1. 
\begin{itemize}
    \item [(i)] IH1 and a) implies that $m+1 = i_k$. This along with b) implies that $m+1 = i_{k'-1}$. Hence proving (i) for Case 1.
\item [(ii)] IH2 along with c) and d) implies that $\delta_\seq((q,k), w'[m+1]) = \delta(q, w[m+1]) \times \{k+1\}$. Hence, if $(q',k') \in A_\seq<w', m+1>$ then $q' \in A<w'', m+1>$. Hence, there exists a $w = w''$ such that $q' \in A<w, m+1>$, proving (ii) for Case 1.
\end{itemize}

\noindent \textbf{Case 2} $m+1 \in \boundaryint(w')$ : This implies 
\begin{itemize}
    \item [(1)] $m+1 \notin \{i_1, i_2, \ldots, i_k\}$.
    \item[(2)] $k' = k$ (by construction of $\delta_\seq$). 
    \item[(3)] $w''[m+1]\subseteq \Sigma $. In other words, $m+1$ is either an unrestricted point in both $w''$.
\end{itemize}
Now we have 
\begin{itemize}
    \item [(i)]  IH1 implies $i_{k-1} \le m < m+1 \le i_{k}$. This along with 1) and 2) implies $i_{k'-1} \le m < m+1 < i_{k'}$. Hence proving (i) for Case 2.
\item [(ii)] IH2 along with 3) and the construction of $\delta_\seq$ implies $\delta_\seq((q,k), w'[m+1]) = (\bigcup \delta_\seq((q,k), w'[m+1] \cup \{J\}) \cup \delta_\seq((q,k), w'[m+1]))\times k$ for $J \in I_\nu$ such that there exists $j<k'<l$ such that $\seq[j] = \seq[l] = J$. 
\end{itemize}

Hence, $J$ is an interval which appears twice in $\seq$ and only one of those $J$'s have been encountered within first $m$ letters. Hence, the prefix $w'[1...m+1]$ and suffix $w'[m+2...]$ contains exactly one $J$ time restricted point. 
This implies that \\

\noindent (Case A) $q' \in \delta_\seq((q,k), w'[m+1] \cup \{J\})$ for some $J$ such that $w'[1...m]$ and $w'[m+2...]$ contains exactly one $J$- time restricted point or \\
\noindent  (Case B) $q' \in \delta_\seq((q,k), w'[m+1])$.\\

As $\normalize(w'') = w'$, first and last $J$ time restricted points are the same in both $w''$ and $w'$. Hence, first $J$-time restricted point in $w''$ is within $w''[1...m]$ and the last is within $w''[m+2...]$. Consider a set of words $W$ such that for any $w \in W$, $w[1...m] = w''[1...m]$, $w[m+2...] = w'[m+2...]$ and either $w[m+1] = w'[m+1]$ or $w[m+1] = w'[m+1] \cup \{J\}$ where $J \in I_\nu$ such that both $w'[1...m]$ and $w'[m+2...]$ contains $J$-time restricted points. Notice that $m+1 \notin \boundaryint(w)$. Hence, making it time unrestricted will still imply $\boundaryint(w) = \boundaryint(w')$. 

When there exists a $J$ restricted time point in prefix $w[1....m]$ and suffix $w[m+2...]$ for $J \in I_\nu$, making point $m+1$ as $J$ restricted time point will still imply $\boundaryint(w) = \boundaryint(w')$. Hence, this implies that $\normalize(w) = w'$ for any $w \in W$. Moreover, as for any $w \in W$, $w[1...m] = w''[1...m]$, $A<w,m> = A<w'',m>$ and $A_\seq<w,m> = A_\seq<w'',m>$. Hence, for any $q \in Q$ such that $(q,k) \in A_\seq<w',m>$ implies for every $w \in W$, $q \in A<w,m>$.

It suffices to show that there exists a $w\in W$ such that $q' \in A<w,m+1>$. In case of Case A, for any word $w \in W$ such that $w[m+1]$ is a $J$-time restricted point $q' \in A<w,m+1>$. Note that such a word exists as Case A implies that $w'[1...m]$ and $w'[m+2...]$ contains exactly one $J$- time restricted point. In case of B, for any word $w  \in W$ where $w[m+1] \subseteq \Sigma$, $q' \in A<w,m+1>$. Hence, proving  for Case 2.
\end{proof}

The  words in $L(A_{\seq})$ are all normalized, and have at most $2|I_\nu|+1$ time restricted points. Thanks to this, its corresponding timed language can be expressed using $\pnregmtl$ formulae with arity at most $2|I_\nu|$.
\smallskip 

\textbf{Reducing NFA of each type to $\pnregmtl$}:
Next, for each $A_\seq$ we construct $\pnregmtl$ formula $\phi_\seq$  such that, for a timed word $\rho$ with $i \in dom(\rho), \rho, i \models \phi_\seq$ iff $\rho, i {\in} \mathsf{Time}(L(A_\seq))$. For any NFA $N = (St,\Sigma, i,Fin,\Delta)$, $q \in Q$ 
$F' \subseteq Q$,  let $N[q,F'] = (St,\Sigma, q, F', \Delta)$.
For the sake of brevity, we denote $N[q,\{q'\}]$ as $N[q,q']$.
We denote by $\rev(N)$, the NFA $N'$ that accepts the reverse of $L(N)$. 
The right/left concatenation of $a \in \Sigma$ with $L(N)$ is denoted  $N\cdot a$ and  $a \cdot N$ respectively.  
 \begin{lemma}
\label{lem:nfaseqtopnregmtl}
We can construct a $\pnregmtl$ formulae $\phi_\seq$ with $\mathsf{Constraint}(\phi_\seq) \subseteq I_\nu$ such that $\rho, i \models \phi_\seq$ iff $\rho, i \in \mathsf{Time}(L(A_\seq))$. 
\end{lemma}
\begin{proof}
\begin{figure}
    \centering\scalebox{0.825}{
    \includegraphics{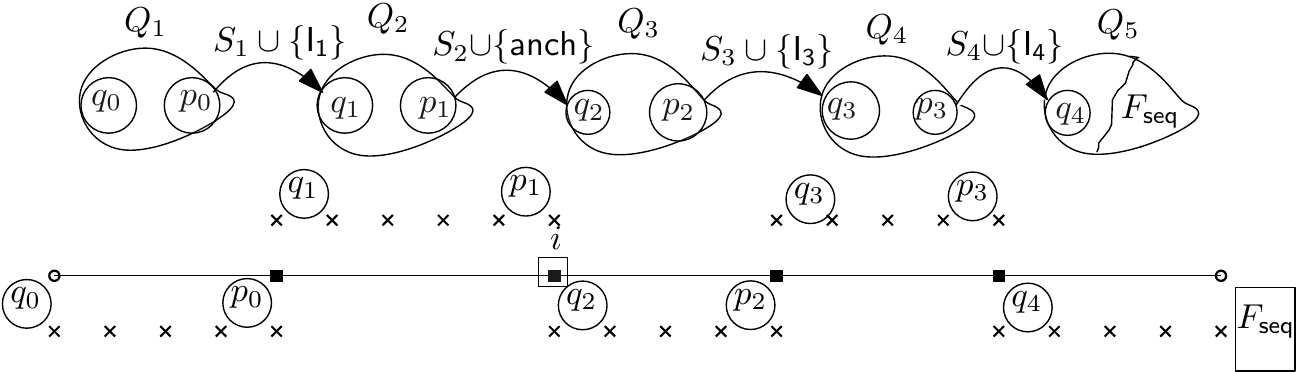}}
    \caption{Figure representing set of runs $A_{\mathsf{I_1 \anch I_3 I_4}}$ of type $R_{Qseq}$ where $Qseq = T_1 T_2 T_3 T_4$, where for $1\le i \le 4$, $T_i = (p_{i-1} \stackrel{S_i\cup \{I_i\}}{\rightarrow} q_i),$ $I_2 = \{\anch\}$}
    \label{fig:nfatopnregmtl}
\end{figure}
Let $\seq{=}I_1\ I_2\ \ldots\ I_n$, and  $I_j{=}\anch$ for some $1{\le}j {\le}n$. 
 Let $\Gamma=2^{\Sigma}$ and
 \\$\Qseq{=}T_1\ T_2\ \ldots T_n$ be a sequence of transitions of $A_{\seq}$ where for any $1 \le i \le n$,  $T_i{=}p_{i-1} \stackrel{S'_{i}}{\rightarrow} q_{i}$, $S'_i=S_i \cup \{I_i\}$, $S_i \subseteq \Sigma$, 
$p_{i-1} \in Q \times \{i-1\}$, $q_{i} \in Q \times \{i\}$. Let $q_0{=}(\init, 1)$.
We define $\mathsf{R}_{\Qseq}$ as set of accepting runs containing transitions $T_1\ T_2\ \ldots T_n$. 
Hence the runs in $\mathsf{R}_{\Qseq}$ are of the following form:
\\
{$T_{0,1}~T_{0,2} \ldots T_{0,m_0}~T_1~~T_{1,1}~ \ldots T_{1,m_1}~T_{2}~~\cdots\cdots~~T_{n{-}1,1}~T_{n{-}1,2} \ldots T_{n}~~T_{n,1} \ldots  T_{n{+}1}\\$}
where the source of the transition $T_{0,1}$ is $q_0$ and the target of the transition $T_{n+1}$ is any accepting state of $A_\seq$. Moreover, all the transitions $T_{i,j}$ for $0\le i\le n$, $1\le j \le n_i$ are of the form $(p' \stackrel{S_{i,j}}{\rightarrow} q')$ where $S_{i,j} \subseteq \Sigma$ and $p',q' \in Q_{i+1}$. Hence, only $T_1, T_2, \ldots T_n$ are labelled by any interval from $I_\nu$. Moreover, only on these transitions the the counter (second element of the state) increments.
Let $\mathcal{W}_{Qseq}$ be set of words associated with any run in $\mathsf{R}_{Qseq}$. Refer figure \ref{fig:nfatopnregmtl} for illustration.
 $w \in W_{Qseq}$ if and only if $w \in L(\re_1).S'_1. L(\re_2).S'_2. \cdots. L(\re_n).S'_n.L(\re_{n+1})$ where 
\\$\re_i = (Q_{i}, 2^{\Sigma}, q_{i-1}, \{p_{i-1}\}, \delta_\seq) \equiv A_{\seq}[q_{i-1},p_{i-1}]$ for $1 \le i \le n$ and \\$\re_{n+1}{=} (Q_{n+1}, 2^{\Sigma}, q_{n}, F_{\seq},\delta_{\seq}){\equiv}A[q_n,F]$.  
Let $\re'_{k}{=}S_{k-1}\cdot{\re_{k}}\cdot S_k$ for $1{\le}k{\le}n{+}1$, with $S_0=S_{n+1}=\epsilon$. 
Let $\rho{=}(b_1, \tau_1) \ldots (b_m, \tau_m)$ be a timed word over $\Gamma$. Then $\rho,i_j \in \mathsf{Time}(W_{Qseq})$ iff $\exists$  $0{\le} i_1 {\le} i_2 {\le} \ldots {\le} i_{j-1} {\le} i_j {\le} i_{j+1} {\le} \ldots {\le} i_n {\le} m$ s.t. 
\\$\bigwedge\limits_{k =1}^{j-1}[(\tau_{i_k}{-}\tau_{i_j} \in I_k)
  \wedge \mathsf{Seg^-}(\rho, i_{k+1}, i_{k},\Gamma) \in L(\rev({\re'_{k}}))] \wedge \bigwedge \limits_{k{=}j}^{n}[(\tau_{i_k}{-}\tau_{i_j} \in I_k)
    \wedge \mathsf{Seg^+}(\rho, i_{k}, i_{k+1},\Gamma) \in L(\re'_k)]$, where $i_0{=}0$ and $i_{n+1}{=}m$.
Hence, by semantics of $\fregk$ and $\sregk$ modalities, $\rho,i \in \mathsf{Time}(\mathcal{W}_{Qseq})$ if and only if $\rho, i{\models} \phi_{\qseq}$ where  \\
$\phi_{\qseq}=\sregm^{j}_{I_{j-1},\ldots,I_{1}} (\rev(\re'_1),\ldots,\rev(\re'_j))(\Gamma) \wedge \fregm^{n-j}_{I_{j+1},\ldots,I_{n}} (\re'_{j+1},\ldots,\re'_{n+1})(\Gamma)$. \\Let $\mathsf{State{-}seq}$ be set of all possible sequences of the form $\Qseq$.
As $A_\seq$ accepts only words which has exactly $n$ time restricted points, the number of possible sequences of the form $\Qseq$ is bounded by $|Q|^{n}$. Hence any word $\rho, i \in \mathsf{Time}(L(A_\seq))$ iff $\rho, i \models
\phi_{\mathsf{seq}}$ where $\phi_{\mathsf{seq}} =     
\bigvee \limits_{\mathsf{qseq} \in \mathsf{State{-}seq}} \phi_{\mathsf{qseq}}$. 
Disjuncting over all possible sequences $\seq{\in}\mathcal{T}(I_\nu)$ we get the required PnEMTL
formula $\phi$.  Moreover, the timing intervals appearing in all the $\fregk$ subformulae of $\phi$ are from $I_\nu$. Similarly, the timing intervals appearing in all the $\sregk$ formulae are from $I^{-}_\nu = \{\langle -l, -u \rangle | \langle l,u \rangle \in I_\nu \wedge l<u\le 0 \}$. If $\I$ is non adjacent, then its intersection closure $I_\nu^-$, $I_\nu$, $I^+_\nu$  are also non adjacent. Hence, if $\I$ is non adjacent then $\phi$ is a non adjacent PnEMTL formula.
\end{proof}
Disjuncting over all possible sequences $\seq \in \mathcal{T}(I_\nu)$ we get the following lemma \ref{lem:nfatopnmtl}.

As a consequence of \ref{lem:toataafa}, \ref{lem:afanfa} and \ref{lem:nfatopnmtl} any (Non Adjacent) 2-Way 1-ATA-rfl of reset depth 0 can be reduced to an equivalent (Non Adjacent) PnEMTL formulae.

\subsubsection{Induction}
\label{app:toata-pnemtl}
Assume that the lemma \ref{lem:toatapnemtl} holds if $\A$ has reset depth less than $n$. Let reset depth of $\A = (\Gamma, Q^+, Q^-, init, \top ,\bot, \Delta, \G)$ be $n$. Let $Q_0,\ldots, Q_m$ be the set of islands of $\A$ with header locations $q^r_{0}, \ldots q^r_{m}$, respectively. Let $Q_0$ be the initial island. Let. A sub automata $\A[q^r_{j}]$ of $\A$ is a 2-Way 1-ATA-rfl same as $\A$ but with initial state $q^r_{j}$ for any $1\le j \le m$. Note that, for any $1\le j \le m$, as $q^r_{j} \prec init$, all the states reachable from $q$ are the states within island $Q_j$ and all the islands lower than $Q_j$. Hence, the reset depth of any subautomata of $\A$ is less than $n$. by induction hypothesis, we can construct a PnEMTL formulae $\varphi_j$ equivalent to $\A[q^r_j]$ for $1\le j \le m$.
Let $W = \{b_1, \ldots b_m\}$ be witness variables for $\A[q^r_j]$ and $\varphi_j$.
We now construct an automata $\A'$ from $\A$ with transition function $\delta'$, set of locations in $Q_0$ and over symbols in $\Gamma \times \{0,1\}^m$ where $j^{th}$ component of the bit vector encodes the truth value of witness $b_j$.
For any $q \in Q_0$, $a\in \Sigma$, $g\in \G$, let $\delta'(q,a,g)$ is a boolean expression constructed from $\delta(q,a,g)$ by replacing all the occurrences of $x.q^r_j$ with truth value of $b_j$. Hence, whenever $b_j$ is false the conjunction of transitions calling $x.s^r_j$ is vanishes in $\delta'$. Note that automata $\A'$ is a reset free automata (as all the literals reset construct are replaced with either 0 or 1). As shown for the base case, we can construct a PnEMTL formulae $\varphi'$ equivalent to $\A'$ over extended alphabets. For any $a\in \Gamma$ and $b \in \{0,1\}^k$ we replace occurrence of $(a,b)$ in $\varphi'$ with $a \wedge \bigwedge \limits_{b(i){=}1} \varphi_i \wedge \bigwedge \limits_{b(j){=}0} \neg \varphi_j$. Hence, By replacing the witnesses with their corresponding formulae we get the required formulae $\varphi$ equivalent to $\A$. Moreover, note that if $\A$ is non adjacent all its sub automata and $\A'$ are non adjacent. Then, by induction hypothesis and by construction of  PnEMTL for reset free automata $\varphi$ is a non adjacent PnEMTL formula.

\subsection{Proof of Lemma \ref{lem:pnemtlqkmso}}
\label{app:pnemtlqkmso}
\textbf{Lemma \ref{lem:pnemtlqkmso}}. (PnEMTL$\subseteq$GQMSO).Given any (NA)PnEMTL formula $\varphi$, we can construct an equivalent (NA)GQMSO formula $\psi$.
\textbf{Proof}. We apply induction on modal depth of the given formula $\varphi$. 
\\\textbf{Base Case}: For modal depth 0, $\varphi$ is a propositional formula and hence it is trivially a AF-GQMSO formula.
Let $\varphi$ be a modal depth 1 formula of the form $\fregk_{I_1, \ldots, I_k}(\re_1, \ldots, \re_{k+1})(\Sigma)$. The reduction for $\sregk$ modality is identical. Moreover, dealing with boolean operators is trivial as the AF-GQMSO are closed under boolean operations. 
Let $\re_j = (2^{\Sigma}, Q_j, init_j, F_j, \delta_j)$. By semantics, for any timed word $\rho = (a_1, \tau_1) \ldots (a_m,\tau_m)$ and $i_0 \in dom(\rho')$, $\rho,i_0 \models \varphi$ iff ${\exists} {i_0{ {<} }i_1{<} i_2 \ldots {<} i_k {<} n}$ s.t.\\ $\bigwedge \limits_{w{=}1}^{k}{[(\tau_{i_w} {-} \tau_{i_0} \in I_w)} \wedge \mathsf{Seg^+}(\rho, i_{w{-}1}+1, i_w,\Sigma) {\in} L({\re_w})].$ 
By BET theorem, we can construct an MSO[$<$] formula $\psi_w(i_{w-1}, i_{w})$  equivalent to condition $\mathsf{Seg^+}(\rho, i_{w{-}1}+1, i_w,\Sigma) {\in} L({\re_w})$. Note that replacing conditions $\mathsf{Seg^+}(\rho, i_{w{-}1}+1, i_w,\Sigma) {\in} L({\re_w})$ with $\psi_i$ will result in a AF-GQMSO formula. Moreover if $\varphi$ is non adjacent then the resulting AF-GQMSO formula is also non adjacent.
We assume that the lemma holds for all the PnEMTL formulae of modal depth $<n$. Let $\varphi = \fregk_{I_1, \ldots, I_k}(\re_1, \ldots, \re_{k+1})(\Sigma \cup S)$ of modal depth $n$. Therefore, $S$ is a set of PnEMTL formula with modal depth $<n$. We replace all the subformulae in $S$ by a witness propositions getting a formula $\varphi'$ of modal depth 1. As with the base case, we can construct an AF-GQMSO formula $\psi'$ equivalent to $\varphi'$. By inductive hypothesis every subformulae $\varphi_i$ in $S$ can be reduced to an equivalent AF-GQMSO formula $\psi_i$. We replace all the witnesses of $\varphi_i$ by $\psi_i$ getting an equivalent formulae $\psi$ over $\Sigma$.
Note that if formula $\varphi_i$ in $S$ are non adjacent then, by induction hypothesis, equivalent $\psi_i$ are in NA-GQMSO formula. Similarly, if $\varphi'$ is NAPnEMTL formula then $\psi'_i$ is NA-GQMSO formula. Hence, if $\varphi$ in non adjacent then equivalent formula $\psi$ is non adjacent too.

\subsection{GQMSO to PnEMTL: An Example}
\label{app:gqmsored}
\begin{example}
In this example, we write a regular expression, in place of NFA wherever required, for the sake of succinctness and readability. Consider a GQMSO formulae $\psi(t) = \texists t_1 \in t+(0,1) \texists t_2 \in t+(-1,0) \psi_{even,b}(t,t_1) \wedge \psi_{odd,a}(t,t_2)$, where $\psi_{even,b}(x,y)(\psi_{odd,a}(x,y))$ is an MSO[$<$] formula which is true iff the number of $b$'s ($a$'s, respectively) strictly between $x$ and $y$ (excluding both) is even (odd, respectively). The regular expression of the behaviour starting from the beginning would be of the form: $\mathbf{(a+b)^* \cdot \{(a+b), x \in (-1,0)\}\cdot(b^*.a.b^*.a.b^*)^*\cdot a\cdot b^*)\cdot} \\\mathbf{anch\cdot (a^*\cdot b\cdot a^*\cdot b\cdot a^*)\cdot \{(a+b), x \in (0,1))\}\cdot (a_b)^*}$. By PnEMTL semantics, $\varphi{=} \fregm^1_{(0,1)}[(a^*.b.a^*.b.a^*),(a+b)^+](\{a,b\}\wedge \sregm^1_{(0,1)}[(b^*.a.b^*.a.b^*)^*.a.b^*),(a+b)^+](\{a,b\})$ when asserted on a point $t$ will accept the same set of behaviours.
\end{example}